\documentclass[final]{IEEEtran}
\usepackage{cite}
\usepackage{threeparttable}
\usepackage{graphicx}
\usepackage{picinpar}
\usepackage[cmex10]{amsmath}
\usepackage{amsmath,amsfonts,amssymb}
\usepackage{subfigure}
\usepackage{enumerate}
\usepackage{color}
\usepackage{cite}
\usepackage{threeparttable}
\usepackage{graphicx}
\usepackage{picinpar}
\usepackage[cmex10]{amsmath}
\usepackage{amsmath,amsfonts,amssymb}
\usepackage{subfigure}
\usepackage{algorithm}
\usepackage{algorithmic}
\usepackage{stfloats}
\usepackage{bm}
\usepackage{amsthm}
\usepackage{changebar}
\usepackage{mathrsfs}
\usepackage{balance}

\begin{document}
\newtheorem{lemma}{Lemma}
\newtheorem{corol}{Corollary}
\newtheorem{theorem}{Theorem}
\newtheorem{proposition}{Proposition}
\newtheorem{problem}{Problem}
\newtheorem{definition}{Definition}
\newcommand{\e}{\begin{equation}}
\newcommand{\ee}{\end{equation}}
\newcommand{\eqn}{\begin{eqnarray}}
\newcommand{\eeqn}{\end{eqnarray}}
\title{ Spatially Common Sparsity Based Adaptive Channel Estimation and Feedback for FDD Massive MIMO}
\author{Zhen Gao,~\IEEEmembership{Student Member,~IEEE}, Linglong Dai,~\IEEEmembership{Senior Member,~IEEE},\\
 Zhaocheng Wang,~\IEEEmembership{Senior Member,~IEEE},
 Sheng Chen,~\IEEEmembership{Fellow,~IEEE}%
\thanks{Z. Gao, L. Dai and Z. Wang (E-mails: gaozhen010375@foxmail.com;
 daill@tsinghua.edu.cn; zcwang@tsinghua.edu.cn) are with Department of Electronic
 Engineering, Tsinghua University, Beijing 100084, China.} %
\thanks{S. Chen (E-mail: sqc@ecs.soton.ac.uk) is with Electronics and Computer Science,
 University of Southampton, Southampton SO17 1BJ, UK, and also with King
 Abdulaziz University, Jeddah 21589, Saudi Arabia.} %
  %
%\vspace{-10mm}
}

\maketitle

\begin{abstract}
 This paper proposes a spatially common sparsity based adaptive channel estimation
 and feedback scheme for frequency division duplex based massive multi-input multi-output
 (MIMO) systems, which adapts training overhead and pilot design to reliably estimate and
 feed back the downlink channel state information (CSI) with significantly reduced overhead.
 Specifically, a non-orthogonal downlink pilot design is first proposed, which is very
 different from standard orthogonal pilots. By exploiting the spatially common sparsity
 of massive MIMO channels, a compressive sensing (CS) based adaptive CSI acquisition
 scheme is proposed, where the consumed time slot overhead only adaptively depends on the
 sparsity level of the channels. Additionally, a distributed sparsity adaptive matching
 pursuit algorithm is proposed to jointly estimate the channels of multiple subcarriers.
 Furthermore, by exploiting the temporal channel correlation, a closed-loop channel
 tracking scheme is provided, which adaptively designs the non-orthogonal pilot according
 to the previous channel estimation to achieve an enhanced CSI acquisition. Finally,
 we generalize the results of the multiple-measurement-vectors case in CS and derive the
 Cramer-Rao lower bound of the proposed scheme, which enlightens us to design the
 non-orthogonal pilot signals for the improved performance. Simulation results demonstrate
 that the proposed scheme outperforms its counterparts, and it is capable of approaching
 the performance bound.
\end{abstract}

\begin{IEEEkeywords}
 Massive multi-input multi-output, frequency division duplex, compressive sensing, %
 channel estimation, feedback, spatially common sparsity, temporal correlation
\end{IEEEkeywords}

\IEEEpeerreviewmaketitle

%\vspace{-1mm}
\section{Introduction}\label{S1}

 By exploiting the increased degree of freedom in the spatial domain, massive
 multi-input multi-output (MIMO) can enhance the spectrum efficiency and energy
 efficiency by orders of magnitude \cite{MMIMOover,Scaleup}. To harvest the benefits
 of massive MIMO, the base station (BS) needs the accurate channel state information
 (CSI) in the downlink for beamforming, resource allocation, and other operations.
 However, it is challenging for the BS to acquire the accurate downlink CSI in
 frequency division duplex (FDD) based massive MIMO, since the overhead used for the
 downlink channel estimation and feedback can be prohibitively high. Most of the
 researches sidestep this challenge by assuming the time division duplex (TDD)
 protocol. In TDD based massive MIMO, the CSI in the uplink can be more easily
 acquired at the BS due to the limited number of users, and the channel reciprocity
 property can be exploited to realize the downlink channel estimation using the
 uplink channel estimation \cite{{MMIMOover},{Scaleup},{TDD},{Haifan}}. However, in
 TDD massive MIMO, the CSI acquired in the uplink may not be always accurate for the
 downlink due to the calibration error of radio frequency chains \cite{FDDocloop}.
 In addition, frequency division duplex (FDD) protocol still dominates current
 wireless networks, where the downlink channel estimation is necessary, since the
 channel reciprocity does not hold. Thus, it is of great importance to explore an
 efficient approach to enable massive MIMO to be backward compatible with current
 wireless networks \cite{MIMOOFDM}. In this paper, we focus on the reliable and
 efficient channel estimation and feedback for FDD massive MIMO.

 Channel estimations in small-scale MIMO are usually based on orthogonal pilots
 \cite{{Orthp},{Orthp2},{LTE},{orthogonal}}. In Long Term Evolution-Advanced (LTE-A),
 for example, pilots associated with different BS antennas occupy the different
 frequency-domain subcarriers \cite{LTE}. Pilot signals can be also orthogonal in the
 time or code domain. However, the overhead of orthogonal pilots increases with the
 number of BS antennas, which becomes unaffordable for massive MIMO. For FDD massive
 MIMO, the work \cite{Love} proposed a pilot design for the downlink channel estimation
 by exploiting the channel statistics. However, the acquisition of the downlink channel
 covariance matrix is challenging in practice.
 An open-loop and closed-loop training based channel estimation scheme was proposed in
 \cite{FDDocloop}. Nevertheless, the long-term channel statistics required by the user
 may increase the training time and memory cost. In \cite{MIMOOFDM}, a sparse channel
 estimation scheme was proposed to acquire CSI with significantly reduced pilot overhead
 by exploiting the sparsity of time-domain channel impulse response (CIR). But this
 time-domain sparsity of the channels may not exist when the number of scatterers at
 the user side becomes large. Furthermore, \cite{{Love},{FDDocloop},{MIMOOFDM}} do not
 consider the channel feedback to the BS. In order to obtain the fine-grain spatial
 channel structures, the conventional codebook based CSI feedback schemes may become
 impossible, since the dimension of codebook can be huge in massive MIMO. Hence the design,
 storage, and encoding of the high-dimensional codebook can be difficult \cite{{multi_VTC}}.
 The compressive sensing (CS) based channel feedback schemes for massive MIMO were proposed
 to reduce the feedback overhead by exploiting the spatial correlation of CSI
 \cite{{multi_VTC},{WCNC12}}. However, these schemes do not consider downlink channel
 estimation. By exploiting the spatially joint sparsity of multiple users' channel matrices,
 the works \cite{{Rao1},{Rao2}} proposed a joint orthogonal matching pursuit (OMP) based
 CSI acquisition scheme. However, this scheme cannot adaptively adjust the required overhead
 according to the sparsity level of the channels. Moreover, the spatially joint sparsity
 may disappear when the users are not spatially close. Even for the best case that multiple
 users' channel matrices share the spatially common sparsity, the sparse CSI acquisition
 problem is a multiple-measurement-vectors (MMV) problem, where the reduction in required
 overhead is limited.

 Recent study and experiments have shown that the wireless channels between the BS and users
 exhibit a small angle spread seen from the BS~\cite{{Haifan},{JSDM},{esprit}}. Due to the
 small angle spread and large dimension of the channels, massive MIMO channels exhibit the
 sparsity in the virtual angular domain~\cite{Virtual_rep}. Moreover, since the spatial
 propagation characteristics of the wireless channels within the system bandwidth are nearly
 unchanged, such sparsity is shared by subchannels of different subcarriers when the widely
 used orthogonal frequency-division multiplexing (OFDM) is considered. This phenomenon is
 referred to as the \emph{spatially common sparsity within the system bandwidth} \cite{TSE}.
 Besides, due to the temporal correlation of the channels \cite{TSE}, massive MIMO channels
 are quasi-static in several adjacent time slots or one time block consisting of multiple
 time slots. Moreover, the support set of the sparse channels in the virtual angular domain
 is almost unchanged in multiple time blocks, which is referred to as the \emph{spatially
 common sparsity during multiple time blocks}.

 By exploiting the spatially common sparsity and the temporal correlation of massive MIMO
 channels, this paper proposes an adaptive channel estimation and feedback scheme with low
 overhead. The proposed scheme consists of two stages: a CS based adaptive CSI acquisition
 with the adaptive training overhead and a follow-up closed-loop channel tracking with the
 adaptive pilot design. Specifically, the BS transmits the proposed non-orthogonal pilot.
 The users simply feed back the received non-orthogonal pilot signals to the BS. According
 to the feedback signals, the CS based adaptive CSI acquisition scheme acquires the downlink
 CSI at the BS with the adaptive training time slot overhead. For this stage, a distributed
 sparsity adaptive matching pursuit (DSAMP) algorithm is proposed to acquire the CSI, whereby
 the spatially common sparsity of massive MIMO channels within the system bandwidth is
 exploited. By exploiting the spatially common sparsity of massive MIMO channels during
 multiple time blocks, the closed-loop channel tracking scheme is proposed to track the
 channels in the second stage. For this stage, the BS can adaptively adjust the pilot signals
 according to the previous acquired CSI, and a simple least squares (LS) algorithm is
 used to estimate the channels with improved performance. Additionally, we generalize
 the results for the conventional MMV to the generalized MMV (GMMV) and provide the Cramer-Rao
 lower bound (CRLB) of the proposed scheme, which enlightens us to design the non-orthogonal
 pilot signals. Simulation results verify that the proposed scheme is superior to its
 counterparts, and it is capable of approaching the performance bound. We now summarize our
 novel contributions.
\begin{itemize}
\item \textbf{CS based adaptive CSI acquisition scheme}: By exploiting the spatially
 common sparsity of massive MIMO channels within the system bandwidth, this scheme
 substantially reduces the required time slot overhead for channel estimation and feedback,
 where the required time slot overhead is only dependent on the channel sparsity level,
 rather than the number of BS antennas as in conventional CSI acquisition schemes.
\item \textbf{Closed-loop channel tracking scheme}: By leveraging the spatially common
 sparsity of massive MIMO channels during multiple time blocks, this scheme can further
 reduce the required time slot overhead.
\item \textbf{Non-orthogonal downlink pilot design at BS}: i)~We theoretically prove
 that the GMMV outperforms the conventional MMV on the sparse signal recovery performance.
 This enlightens us to design the non-orthogonal pilot for CS based adaptive CSI acquisition
 for improving channel estimation performance. ii)~We derive the CRLB for the proposed
 scheme. In the stage of closed-loop channel tracking, the derived CRLB enlightens us to
 adaptively design the non-orthogonal pilot according to the previous channel estimation
 for further improving performance.
\item \textbf{DSAMP algorithm}: This algorithm leverages the spatially common sparsity of
 massive MIMO channels to jointly estimate multiple channels associated with different
 subcarriers. Compared with the conventional algorithms, such as sparsity adaptive
 matching pursuit (SAMP), subspace pursuit (SP) and OMP, the proposed DSAMP substantially
 reduces the required time slot overhead with similar computational complexity.
\end{itemize}

 Throughout our discussions, scalar variables are denoted by normal-face letters, while
 boldface lower and upper-case symbols denote column vectors and matrices, respectively,
 and $\textsf{j}=\sqrt{-1}$ is the imaginary axis. The Moore-Penrose inversion, transpose
 and conjugate transpose operators are given by $(\cdot )^{\dag}$, $(\cdot )^{\rm T}$ and
 $(\cdot )^{*}$, respectively, while $\lceil \cdot \rceil$ is the integer ceiling operator
 and $(\cdot )^{-1}$ is the inverse operator. The $\ell_{0}$-norm and $\ell_{2}$-norm are
 given by $\|\cdot\|_0$ and $\|\cdot\|_2$, respectively, and $\left| \Gamma \right|$ is the
 cardinality of the set $\Gamma$. The support set of the vector $\mathbf{a}$ is denoted by
 ${\rm supp}\{\mathbf{a}\}$, $\left[ {\bf a} \right]_i$ denotes the $i$th entry of the
 vector $\mathbf{x}$, and $\left[ {\bf A} \right]_{i,j}$ denotes the $i$th-row and
 $j$th-column element of the matrix $\mathbf{A}$, while ${\bf I}_K$ is the $K\times K$
 identity matrix. The rank of ${\bf A}$ is denoted by ${\rm rank}\{{\bf A}\}$ and
 ${\rm Tr}\{\cdot \}$ is the matrix trace operator, while ${\rm E}\{\cdot \}$ is the
 expectation operator and ${\rm var}\{\cdot \}$ is the variance of a random variable.
 Finally, $\left( {\bf a} \right)_{\Gamma}$ denotes the entries of $\mathbf{a}$ whose
 indices are defined by $\Gamma$, while $\left( {\bf A}\right)_{\Gamma}$ denotes a
 sub-matrix of $\mathbf{A}$ with column indices defined by $\Gamma$.

%\vspace{-1mm}
\section{System Model}\label{S2}

%\vspace{-1mm}
\subsection{Massive MIMO in the Downlink}\label{S2.1}

 In a typical massive MIMO system, the BS employing $M$ antennas simultaneously serves $K$
 single-antenna users \cite{Scaleup}, where $M \gg K$. For the subchannel at the $n$th
 subcarrier, where $1\le n\le N$ and $N$ is the size of the OFDM symbol, the received signal
 $y_{k,n}$ of the $k$th user can be expressed as
\begin{equation}\label{equ:H} % eq1
 y_{k,n} = {\bf h}_{k,n}^{\rm T}{\bf x}_n + w_{k,n} ,
\end{equation}
 where ${\bf h}_{k,n}\in \mathbb{C}^{M\times 1}$ denotes the downlink channel between
 the $k$th user and the $M$ antennas at the BS, ${\bf x}_n \in \mathbb{C}^{M \times 1}$
 is the transmitted signal after precoding, and $w_{k,n}$ is the associated additive
 white Gaussian noise (AWGN). The received signal of the $K$ users
 ${\bf y}_n=\big[y_{1,n} ~ y_{2,n} \cdots y_{K,n}\big]^{\rm T}\! \in \! \mathbb{C}^{K \times 1}$
 can be collected together as
\begin{equation}\label{equ:downlink} % eq2
 {\bf y}_n = {\bf H}_n{\bf x}_n + {\bf w}_n ,
\end{equation}
 in which ${\bf H}_n=\big[{\bf h}_{1,n} ~ {\bf h}_{2,n}\cdots {\bf h}_{K,n}\big]^{\rm T}
 \in \mathbb{C}^{K \times M}$ is the downlink channel matrix, and ${\bf w}_n=\big[
 w_{1,n} ~ w_{2,n} \cdots w_{K,n}\big]^{\rm T}\in \mathbb{C}^{K \times 1}$ is the
 corresponding AWGN vector.

%\vspace{-1mm}
\subsection{Massive MIMO Channels in Virtual Angular Domain}\label{S2.2}

 We model the channel vector ${\bf h}_{k,n}$ by using the virtual angular domain
 representation \cite{{TSE},{Virtual_rep}}
 \begin{equation}\label{equ:channelmode1} % eq3
 y_n = {\bf h}_n^{\rm T} {\bf x}_n + w_n = \widetilde{\bf h}_n^{\rm T}{\bf A}_B^{*}
  {\bf x}_n + w_n ,
\end{equation}
 where the user index $k$ in $y_{k,n}$, ${\bf h}_{k,n}$ and $w_{k,n}$ is dropped to
 simplify the notations, while ${\bf h}_n^{\rm T} =\widetilde{\bf h}_n^{\rm T}{\bf A}_B^{*}$
 and ${\bf A}_B\in \mathbb{C}^{M\times M}$ is the unitary matrix representing the
 transformation matrix of the virtual angular domain at the BS side. ${\bf A}_B$ is
 determined by the geometrical structure of the BS's antenna array.

\begin{figure}[hp!]
%%\vspace{-4mm}
\centering
\includegraphics[width=1\columnwidth, keepaspectratio]{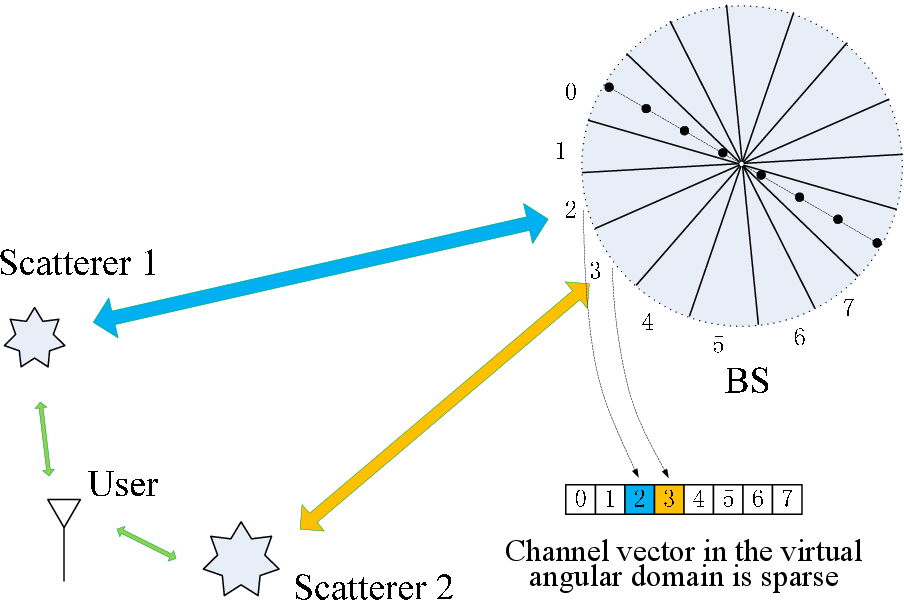}
%%\vspace{-4mm}
\caption{Channel vector representation in the virtual angular domain, where the BS
 employs the ULA with half wave-length spacing, $M=8$, and two clusters of
 scatterers are considered as an example.}
\label{fig:virtual_angle} % Fig. 1
%%\vspace{-4mm}
\end{figure}

 To intuitively explain the channel vector $\widetilde{\bf h}_n$, a simple example is
 illustrated in Fig.~\ref{fig:virtual_angle}, where the BS employs the uniform linear
 array (ULA) with the antenna spacing of $d ={\lambda}/{2}$ and $\lambda$ is the
 wave-length. In this case, ${\bf A}_B$ becomes the discrete Fourier transform (DFT)
 matrix \cite{{TSE}}. The channel vector in the virtual angular domain then simply
 means to `sample' the channel in the angular domain at equi-spaced angular intervals
 at the BS side, or equivalently to represent the channel in the virtual angular domain
 coordinates. More specifically, the $m$th element of $\widetilde{\bf h}_n$ is the
 channel gain consisting of the aggregation of all the paths, whose transmit/receive
 directions are within an angular window around the $m$th angular coordinate \cite{{TSE}}.

 As the BS is usually elevated high with few scatterers around, while users are located
 at low elevation with relatively rich local scatterers, the angle spread at the BS side
 is small \cite{{Haifan},{JSDM},{esprit}}. Since the angle spread is limited at the BS,
 a small part of the elements in $\widetilde{\bf h}_n$ contain almost all the multipath
 signals reflected, diffracted, or refracted by scatterers around the user. If we take
 the typical angular-domain spread of $10^{\circ}$ and the ULA with $M=128$ as an example
 \cite{Haifan}, the uniformly virtual angular domain sampling interval is
 $\varphi_s=180^{\circ}/M=1.406^{\circ}$ \cite{Virtual_rep}, and the vast majority of the
 channel energy is concentrated on around $8=\lceil 10^{\circ}/1.406^{\circ} \rceil$
 virtual angular domain coordinates, which is far smaller than the total dimension $M=128$
 of the channel vector. Consequently, $\widetilde{\bf h}_n$ exhibits the sparsity \cite{TSE},
 namely,
\begin{equation}\label{equ:channelmode3} % eq4
 \left|\Theta_n\right| = \left| {\rm supp}\left\{ \widetilde{\bf h}_n \right\} \right| = S_a \ll M ,
\end{equation}
 where $\Theta_n$ is the support set, and $S_a$ is the sparsity level.

 Moreover, since the spatial propagation characteristics of the channels within the system
 bandwidth (e.g. 10 MHz in typical LTE-A systems) are almost unchanged, the subchannels
 associated with different subcarriers share very similar scatterers in the propagation
 environment~\cite{TSE}. Hence the small angle spreads of the subchannels within the system
 bandwidth are very similar. Consequently, $\left\{\widetilde{\bf h}_n\right\}_{n=1}^{N}$
 have the common sparsity, namely,
 \begin{equation}\label{eq5}
 {\rm supp}\left\{\widetilde{\bf h}_1\right\} = {\rm supp}\left\{\widetilde{\bf h}_2\right\}
 = \cdots = {\rm supp}\left\{\widetilde{\bf h}_N\right\} = \Theta ,
\end{equation}
 which is illustrated in Fig.~\ref{fig:CSI_matrix}.

\begin{figure}[hp!]
%%\vspace{-8mm}
\centering
\includegraphics[width=1\columnwidth, keepaspectratio]{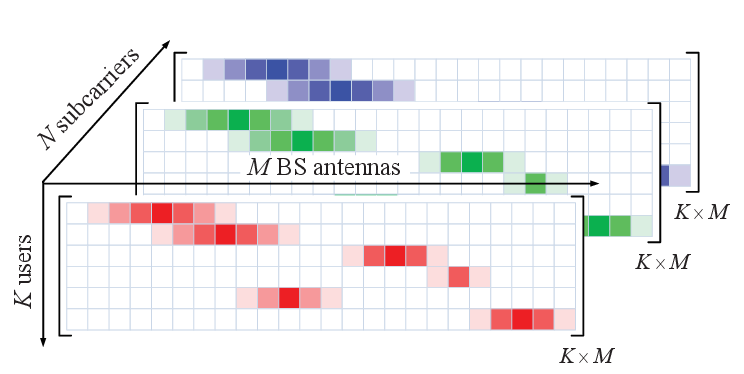}
%%\vspace{-7mm}
\caption{The virtual angular-domain channel vectors within the system bandwidth exhibit
 the common sparsity.}
\label{fig:CSI_matrix} % Fig. 2
%%\vspace{-8mm}
\end{figure}

%%\vspace{-1mm}
\subsection{Temporal Correlation of Wireless Channels}\label{temporal}

\begin{figure*}[tp!]
%%\vspace{-1mm}
\centering
\includegraphics[width=15cm, keepaspectratio]{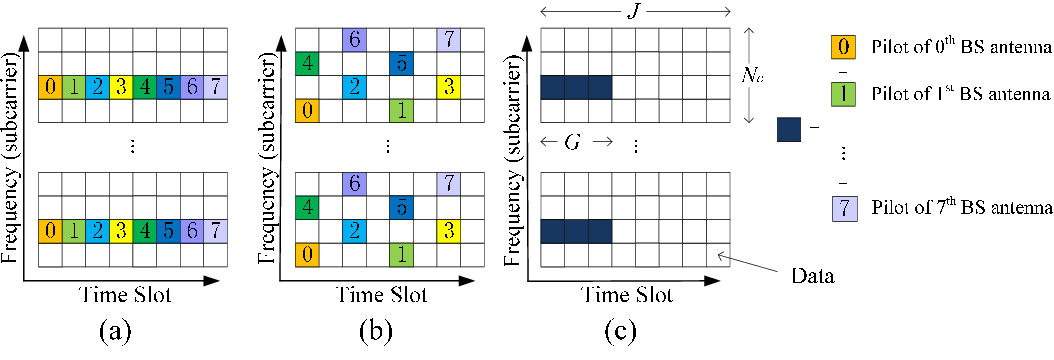}
%%\vspace{-8mm}
\caption{(a)~Time-domain orthogonal pilot \cite{orthogonal}, (b)~time-frequency orthogonal
 pilot in LTE-A \cite{LTE}, and (c)~proposed non-orthogonal pilot, assuming $M=8$.}
 \label{fig:pilot} % Fig. 3
%%\vspace{-7mm}
\end{figure*}
%%\vspace{-1mm}

 Since the user mobility is not very high in massive MIMO systems, the channels remain
 static for the duration of a block that consists of $J$ consecutive time slots, while
 the channels change from block to block. Here, one time slot represents one OFDM symbol.
 This block fading implies that ${\bf h}_n^{(q,t)}={\bf h}_n^{(q)}$ for $1\le t\le J$,
 where ${\bf h}_n^{(q,t)}$ is the channel at the $t$th time slot of the $q$th block
 and ${\bf h}_n^{(q)}$ denotes the quasi-static channel in the $q$th block. Similarly,
 there exists the quasi-static relationship $\widetilde{\bf h}_n^{(q,t)}
 =\widetilde{\bf h}_n^{(q)}$ for $1\le t\le J$, with $\widetilde{\bf h}_n^{(q,t)}$ and
 $\widetilde{\bf h}_n^{(q)}$ being the virtual angular representations of ${\bf h}_n^{(q,t)}$
 and ${\bf h}_n^{(q)}$, respectively.

 For massive MIMO channels, $J < M$ due to the limited coherence time and the large number
 of BS antennas. For example, consider massive MIMO systems with: the carrier frequency
 $f_c=2$\,GHz, the system bandwidth $B_s=10$\,MHz, the OFDM size $N=2048$, the number of
 BS antennas $M=128$, and the maximum delay spread $\tau_{\rm max}=6.4\,\mu{\rm s}$ (need
 the guard interval $N_g=64$) \cite{{Scaleup},{channel_model_for4g}}. Suppose that the
 maximum mobile velocity of the supported users is $v=36$\,km/h. Then the maximum Doppler
 frequency shift is $f_d=vf_c/c=66.67$\,Hz, where $c$ is the velocity of electromagnetic
 wave. Hence the coherence time $T_c=\sqrt{9/\left( 16\pi f_d^2 \right)} \approx 6.3\,{\rm ms}$
 \cite{matlab}, or the coherence time slots $J=T_c B_s/(N+N_g)\approx 30$, which is much
 smaller than $M$.

 Since the channels change from block to block, they must be estimated in every time block,
 which may impose very high complexity and overhead. Fortunately, experiments and theoretical
 analysis have shown that although the channels vary continuously from one block to another,
 the variation rate of the channel angle spread is much lower than that of the associated
 channel gains \cite{Virtual_rep}. This implies that
\begin{equation}\label{equ:common_block} % eq 6
{\rm supp}\!\left\{ \widetilde{\bf h}_n^{(q)}\!\right\} ={\rm supp}\!\left\{ \widetilde{\bf h}_n^{(q+1)}\!\right\}
 \!\!   =  \!\cdots \!= {\rm supp}\!\left\{ \widetilde{\bf h}_n^{(q+Q-1)}\!\right\} ,
\end{equation}
 where $Q$ is the number of consecutive time blocks over the duration of which the
 common support of virtual angular domain channels holds. For the example of
 Fig.~\ref{fig:virtual_angle}, assume that the distance between the BS and the user is
 $L_{BU}=250\,{\rm m}$ and $v=36\,{\rm km/h}$. Further assume the case of the mobile
 direction of the user being perpendicular to the direction connecting the BS and the
 user. Then, over the duration of $Q=5$ successive time blocks, the maximum variance in
 the virtual angular domain is around $\theta_{\Delta}=\arctan\left( Q T_c v/L_{BU}\right)
 \approx 0.072^{\circ}$. Such a small variance of the angle spread is negligible,
 compared to the resolution of the virtual angular domain $\varphi_s=1.406^{\circ}$. If
 $v<36\,{\rm km/h}$ and/or the mobile direction of the user is not perpendicular to the
 direction connecting the BS and the user, $Q$ can be larger than 5.

\subsection{Challenges of Channel Estimation and Feedback}\label{challenging}

 Consider the downlink channel estimation in the $q$th time block. To reliably estimate
 the channel of the $n$th subcarrier, the user should jointly utilize the received pilot
 signals over several successive time slots, say, $G$ time slots, for channel estimation.
 Let $y_n^{(q,t)}$ be the received pilot of (\ref{equ:channelmode1}) at the $n$th
 subcarrier in the $t$th time slot, and $y_n^{(q,t)}$ for $1\le t\le G$ can be collected
 together in the vector ${\bf y}_n^{[q,G]}=\big[ y_n^{(q,1)} ~ y_n^{(q,2)} \cdots
 y_n^{(q,G)}\big]^{\rm T} \in \mathbb{C}^{G\times 1}$. Then
\begin{equation}\label{equ:channelmode7} % eq 7
 {\bf y}_n^{[q,G]} = {\bf X}^{[q,G]}_n {\bf h}_n^{(q)} + {\bf w}_n^{[q,G]} ,
\end{equation}
 where  ${\bf X}^{[q,G]}_n=\big[ {\bf x}_n^{(q,1)} ~ {\bf x}_n^{(q,2)} \cdots
 {\bf x}_n^{(q,G)}\big]^{\rm T}\in \mathbb{C}^{G\times M}$ with ${\bf x}_n^{(q,t)}\in
 \mathbb{C}^{M\times 1}$ being the transmitted pilot signals in the $t$th time slot,
 and ${\bf w}_n^{[q,G]}=\big[ w_n^{(q,1)} ~ w_n^{(q,2)} \cdots w_n^{(q,G)}\big]^{\rm T}\in
 \mathbb{C}^{G\times 1}$ is the corresponding AWGN vector. To accurately estimate the
 channel from (\ref{equ:channelmode7}), the value of $G$ used in conventional algorithms,
 such as the minimum mean square error (MMSE) algorithm, is heavily dependent on the
 value of $M$. Usually, $G$ can be larger than $J$, which leads
 to the poor channel estimation performance~\cite{orthogonal}. Moreover, to minimize the
 mean square error (MSE) of the channel estimate, ${\bf X}_n^{[q,G]}$ should be a unitary
 matrix scaled by a transmit power factor \cite{orthogonal}. Usually, ${\bf X}^{[q,G]}_n$
 is a diagonal matrix with equal-power diagonal elements. Such a pilot design is
 illustrated in Fig.~\ref{fig:pilot}\,(a), which is called the time-domain orthogonal
 pilot. It should be pointed out that in MIMO-OFDM systems, to estimate the channel
 associated with one transmit antenna, $P$ pilot subcarriers should be used, and usually
 $P=N_g$ is considered since $N_c=N/N_g$ adjacent subcarriers are correlated \cite{orthogonal}.
 Hence the total pilot overhead to estimate the complete MIMO channel is $P_{\rm total}=
 P G=N_g M$. Similarly, LTE-A adopts the time-frequency orthogonal pilots as shown in
 Fig.~\ref{fig:pilot}\,(b), which also needs $P_{\rm total}=P M=N_g M$. These two kinds of
 orthogonal pilots are equivalent, since both of them are based on the framework of Nyquist
 sampling theorem and have the same pilot overhead. Hence we only consider the time-domain
 orthogonal pilot in this paper, and we will propose an efficient non-orthogonal pilot scheme.

 Codebook based channel feedback schemes are widely adopted in small-scale MIMO systems.
 However, to obtain the fine-grain spatial channel structures in massive MIMO systems, the
 codebook size can be huge. Moreover, the storage and encoding of large dimension codebook
 at the user is challenging. To overcome this difficulty, we combine the channel estimation
 and feedback, whereby the CSI acquisition is mainly realized at the BS which has sufficient
 computation capability. By exploiting the spatially common sparsity and temporal
 correlation of massive MIMO channels, the proposed scheme can significantly reduce the
 required overhead and complexity for channel estimation and feedback.

%%\vspace{-1mm}
\section{Spatially Common Sparsity Based Adaptive Channel Estimation and Feedback Scheme}\label{S3}

 The procedure of the proposed adaptive channel estimation and feedback scheme is first
 summarized.

\emph{Step 1}: In each time slot, the BS transmits a non-orthogonal pilot to the user,
 and the user directly feeds back the received pilot signal to the BS. Except for %in
 \emph{Step 4}, the pilot signal is designed in advance.

\emph{Step 2}: The BS uses the proposed DSAMP algorithm to jointly reconstruct multiple
 sparse virtual angular domain channels of high dimension from the feedback signals of low
 dimension collected in multiple time slots.

\emph{Step 3}: The BS judges the reliability of the estimated sparse channels according
 to a pre-specified criterion. If the given criterion is met, the BS stops transmitting
 pilot in the following time slots, and the acquired CSI at the BS is used for precoding
 and user scheduling in the current time block. Otherwise, the BS goes back to
 \emph{Step 1} until the feedback signals are sufficient for acquiring the reliable CSI.

\emph{Step 4}: Since the BS has acquired the estimated support set $\widehat{\Theta}$
 and the estimated sparsity level $\widehat{S}_a$,
 it can directly use the LS algorithm to estimate the channels in every time block of
 the following $Q-1$ blocks. Here, the time slot overhead required in \emph{Step 1}
 can be reduced to $G=\widehat{S}_a$, and the pilot signals can be adaptively adjusted
 according to $\widehat{\Theta}$ for further improving performance.

 It is seen that the proposed adaptive channel estimation and feedback scheme consists
 of two stages: the CS based adaptive CSI acquisition in the $q$th time block, which
 includes \emph{Step 1} to \emph{Step 3}, and the following closed-loop channel tracking in
 the following $Q-1$ time blocks, which includes \emph{Step 1} and \emph{Step 4}. We now
 detail all the key technical components.

%%\vspace{-1mm}
\subsection{Non-Orthogonal Pilot for Downlink Channel Estimation}\label{non-orthogonal}

 The proposed non-orthogonal pilot scheme is illustrated in Fig.~\ref{fig:pilot}\,(c).
 Similar to the time-domain orthogonal pilot scheme, $P$ subcarriers are dedicated to
 pilots in each OFDM symbol. However, the proposed scheme allows the non-orthogonal pilot
 signals associated with different BS antennas to occupy the completely identical
 frequency-domain subcarriers.

 The orthogonal pilot based conventional designs usually require $G\ge M$. By contrast,
 the proposed non-orthogonal pilot for CS based adaptive CSI acquisition is capable of
 providing the efficient compression and reliable recovery of sparse signals. Therefore,
 $G$ is mainly determined by $S_a \ll M$. The non-orthogonal pilot of the first stage is
 designed in advance, which will be discussed in Section~\ref{S4.1}. According to the
 CSI acquired in the first stage, the non-orthogonal pilot used for closed-loop channel
 tracking is adaptively designed to minimize both $G$ and the MSE performance of CSI
 acquisition, which will be illustrated in Section~\ref{S3.4}.

 For the placement of pilot subcarriers, the widely used equi-spaced pilot is considered,
 and the specific reason is given in Section~\ref{S4.3}. For convenience, we denote
 $\Omega_{\xi}=\left\{ \xi_1,\xi_2,\cdots ,\xi_{P}\right\}$ as the index set of the
 pilot subcarriers, where $\xi_p$ for $1 \le p \le P$ denotes the subcarrier index
 dedicated to the $p$th pilot subcarrier. It is worth pointing out that the $p$th pilot
 subcarrier is shared by the pilot signals of the $M$ transmit antennas as illustrated
 in Fig.~\ref{fig:pilot}\,(c).

%%\vspace{-1mm}
\subsection{CS Based Adaptive CSI Acquisition Scheme}\label{CSI acquisition}

 In the $q$th time block, as indicated in \emph{Step 1}, the user directly feeds back
 the received pilot signals to the BS without performing downlink channel estimation
 where the feedback channel can be considered as the AWGN channel
 \cite{{multi_VTC},{Rao1},{Rao2},{WCNC12}}. According to (\ref{equ:channelmode7}), at
 the BS, the fed back signal\footnote{The delay of the feedback signal is negligible,
 compared with the relatively long channel coherence time.} (at the $\xi_p$th subcarrier)
 in the $t$th time slot can be expressed as
\begin{equation}\label{equ:channelmode5}  % eq8
 r_p^{(q,t)} = \big(\bar{\bf h}_p^{(q)}\big)^{\rm T} {\bf A}_B^{*} {\bf s}_p^{(q,t)} + v_p^{(q,t)} ,
 ~1\le p\le P,
\end{equation}
 where $r_p^{(q,t)}=y_{\xi_p}^{(q,t)}$ is the $p$th feedback pilot signal in the $t$th
 time slot, $\bar{\bf h}_p^{(q)}=\widetilde{\bf h}_{\xi_p}^{(q)}$ is the virtual angular
 domain channel vector associated with the $p$th pilot subcarrier, ${\bf s}_p^{(q,t)}
 ={\bf x}_{\xi_p}^{(q,t)}$ is the pilot signal vector transmitted by the $M$ BS antennas,
 and $v_p^{(q,t)}=w_{\xi_p}^{(q,t)}$ is the effective noise which aggregates both the
 downlink channel's AWGN and feedback channel's AWGN.

 Due to the quasi-static property of the channel during one time block, the feedback signals
 in $G$ successive time slots can be jointly exploited to acquire the downlink CSI
 at the BS, which can be expressed as
\begin{align}\label{equ:joint_process2} % eq 9
\!\! {\bf r}_p^{[q,G]} \! =\! {\bf S}_p^{[q,G]} \big( {\bf A}_B^{*} \big)^{\rm T} \bar{\bf h}_p^{(q)}
\!  + \!{\bf v}_p^{[q,G]}\! = \!{\bf \Phi}_p^{[q,G]} \bar{\bf h}_p^{(q)} \!+ \!{\bf v}_p^{[q,G]} ,
\end{align}
 for $1\le p \le P$, where ${\bf r}_p^{[q,G]}=\big[ r_p^{(q,1)} ~ r_p^{(q,2)} \cdots
 r_p^{(q,G)}\big]^{\rm T}$, ${\bf S}_p^{[q,G]}=\big[ {\bf s}_p^{(q,1)} ~ {\bf s}_p^{(q,2)}
 \cdots {\bf s}_p^{(q,G)}\big]^{\rm T}\in \mathbb{C}^{G\times M}$, ${\bf v}_p^{[q,G]}=
 \big[ v_p^{(q,1)} ~ v_p^{(q,2)} ~ \cdots v_p^{(q,G)}\big]^{\rm T}$ and ${\bf \Phi}_p^{[q,G]}
 ={\bf S}_p^{[q,G]} \big( {\bf A}_B^{*} \big)^{\rm T}\in \mathbb{C}^{G\times M}$. The system's
 signal noise ratio (SNR) is defined as $\mbox{SNR}={\rm E}\Big\{\Big\|{\bf \Phi}_p^{[q,G]}
 \bar{\bf h}_p^{(q)}\Big\|_2^2\Big\} \Big/ {\rm E}\Big\{\Big\| {\bf v}_p^{[q,G]}\Big\|_2^2\Big\}$,
 according to (\ref{equ:joint_process2}). By exploiting the spatially common sparsity within
 the system bandwidth, the proposed DSAMP algorithm can reconstruct the sparse angular domain
 channels of multiple pilot subcarriers, as will be detailed in Section~\ref{DSAMP}.

\begin{algorithm}[tbp]
\renewcommand{\algorithmicrequire}{\textbf{Input:}}
\renewcommand\algorithmicensure {\textbf{Output:} }
\caption{CS Based Adaptive CSI Acquisition Scheme}
\label{alg1} % Alg 1
\begin{algorithmic}[1]
\STATE Determine the initial time slot overhead $G_0$, and set the iteration index $i=0$.
\REPEAT
\STATE Collect ${\bf r}_p^{[q,G_i]}$ and ${\bf \Phi}_p^{[q,G_i]}$ in (\ref{equ:joint_process2})
 for given $G_i$, $1\le p \le P$. {\scriptsize \%~$G_i$ is the required overhead at the
 $i$th iteration.}
\STATE Acquire the channel vectors $\widehat{\bar{\bf h}}_p^{(q)}$ $\forall p$ by using
 the proposed DSAMP algorithm (Algorithm~\ref{alg:Framwork}).
\STATE $G_{i+1} = G_i + 1$; $i = i+1$.
\UNTIL {$\sum\limits_{p=1}^{P}\!\left\| {\bf r}_p^{[q,G_{i-1}]}\!-\!{\bf \Phi}_p^{[q,G_{i-1}]}\widehat{\bar{\bf h}}_p^{(q)}
 \!\right\|_2^2\!\!/\!({P}G_{i-1})\!\!\le\! \varepsilon$}. {\scriptsize \%~If the error is smaller than the
 threshold $\varepsilon$, end repeat; otherwise, continue transmitting the pilot in the next time~slot.}
\STATE $G_0=G_i-1$. ${\kern 5pt}$ {\scriptsize \%~Optional, determine the initial time slot
 overhead for the next CS based adaptive CSI acquisition.}
\end{algorithmic}
\end{algorithm}

 For practical massive MIMO systems, the sparsity level $S_a$ of the channels in the
 virtual angular domain can be time-varying. If $S_a$ is relatively small, a small time
 slot overhead $G$ is sufficient to acquire an accurate CSI estimate, while if $S_a$ is
 relatively large, a large $G$ is required to guarantee the reliable sparse signal recovery.
 We propose the CS based adaptive CSI acquisition as presented in Algorithm~\ref{alg1},
 which can adaptively adjust $G$ to acquire the reliable CSI at the BS efficiently.

 At the first CS based adaptive CSI acquisition, we need to empirically determine the
 initial time slot overhead $G_0$. Since the typical angle spread is about $10^{\circ}$
 \cite{Haifan}, for massive MIMO with $M=128$, the effective sparsity level $S_a = 8$.
 Thus, we may set $G_0=8$ to start. Given $G_i$, the DSAMP algorithm (Algorithm~\ref{alg:Framwork})
 acquires the set of channel vectors $\widehat{\bar{\bf h}}_p^{(q)}$ $\forall p$.
 If $\sum\nolimits_{p=1}^{P}\left\|
 {\bf r}_p^{[q,G_{i}]}-{\bf \Phi}_p^{[q,G_{i}]}{\widehat{\bar{\bf h}}_p}^{(q)} \right\|_2^2/(PG_i)$
 is larger than the predefined threshold $\varepsilon$, the sparse signal recovery is judged
 to be unreliable. Hence, the training time slots increase by one, and a set of the feedback
 pilot signals and transmitted pilot signals, $r_p^{(q,G_{i+1})}$ and ${\bf s}_p^{(q,G_{i+1})}$
 $\forall p$, are collected in the $(G_{i+1})$th time slot, which are combined with the
 previously collected ${\bf r}_p^{[q,G_i]}$ and ${\bf \Phi}_p^{[q,G_i]}$ to enlarge the
 dimension of the measurement vectors sequentially, yielding
\begin{align*}
 & {\bf r}_p^{[q,G_{i+1}]} \!= \!\! \left[ \!\! \begin{array}{c} {\bf r}_p^{[q,G_i]} \\ r_p^{(q,G_{i+1})}
 \end{array} \!\!\! \right] \! \! \mbox{ and }\!
 {\bf \Phi}_p^{[q,G_{i+1}]} = \!\! \left[ \!\!\! \begin{array}{c} {\bf \Phi}_p^{[q,G_i]} \\
 \Big( {\bf s}_p^{(q,G_{i+1})} \Big)^{\rm T} \big( {\bf A}_B^{*} \big)^{\rm T} \end{array} \!\!\! \right] \!
\end{align*}
 to improve the channel estimation. Furthermore, an appropriate initial time slot overhead
 for the next CS based adaptive CSI acquisition is automatically determined at the end.

%%\vspace{-1mm}
\subsection{Proposed DSAMP Algorithm for Channel Estimation}\label{DSAMP}

 Given the measurements (\ref{equ:joint_process2}), the CSI can be acquired by solving the
 following optimization
\begin{align}\label{equ:target_func} % eq 10
\begin{array}{l}
 \min\limits_{\bar{\bf h}_p^{(q)}, 1\le p\le P} \Big( \sum\nolimits_{p=1}^{P}
 \left\| \bar{\bf h}_p^{(q)} \right\|_0^2 \Big)^{1/2} \\
 {\kern 15pt} {\rm s.t.} {\kern 2pt} {\bf r}_p^{[q,G]} = {\bf{\Phi}}_p^{[q,G]} \bar{\bf h}_p^{(q)},~{\forall p}
  \,~~{\rm and} ~ \left\{ \bar{\bf h}_p^{(q)} \right\}_{p=1}^{P}  \\ {\kern 30pt}
 \mbox{ share the common sparse support set.}
\end{array}
\end{align}
 The DSAMP algorithm, listed in Algorithm~\ref{alg:Framwork}, is used to solve the
 optimization (\ref{equ:target_func}) to simultaneously acquire multiple sparse channel
 vectors at different pilot subcarriers. This algorithm is developed from the
 SAMP algorithm \cite{SAMP}. Specifically, for each stage with the fixed sparsity level
 $\cal{T}$, {\it line 8} selects the $\cal{T}$ potential non-zero elements;  {\it line 9}
 estimates the values associated with the support set $\Omega^{i-1} \cup \Gamma$ using LS;
 {\it line 10} selects $\cal{T}$ most likely supports. {\it Lines~7}$\sim${\it  12} and
 {\it line~21} together aim to find $\cal{T}$ virtual angular domain coordinates which
 contain most of the channel energy. In particular, {\it Lines 7$\sim$12} remove wrong
 indices added in the previous iteration and add the indices associated with the potential
 true indices. If \emph{line 18} is triggered, the algorithm updates $\cal{T}$ and begins
 a new stage. The algorithm is halted when the stopping criteria, indicated in
 \emph{lines~14}$\sim$\emph{17} and \emph{line 23}, are met.

\begin{algorithm}[!tp]
{\small
\renewcommand{\algorithmicrequire}{\textbf{Input:}}
\renewcommand\algorithmicensure {\textbf{Output:} }
\caption{Proposed DSAMP Algorithm}
\label{alg:Framwork} % Alg 2
\begin{algorithmic}[1]
\REQUIRE
 Noisy feedback signals ${\bf r}_p^{[q,G]}$ and sensing matrices ${\bf{\Phi}}_p^{[q,G]}$
 in (\ref{equ:target_func}), $1\le p \le P$; termination threshold $p_{\rm th}$.
\ENSURE
 Estimated channel vectors in the virtual angular domain at multiple pilot subcarriers
 $\widehat{\bar{\bf h}}_p^{(q)}$, $\forall p$.
\STATE ${\cal{T}} = 1$; $i = 1$; $j =1$. {\scriptsize \%~$\cal{T}$, $i$, $j$ are the sparsity
 level of the current stage, iteration index, and stage index, respectively.}
\STATE ${\bf c}_p\!=\!{\bf t}_p\! ={\bf c}_p^{\rm{last}}\!=\!{\bf{0}}\in \mathbb{C}^{M\times 1}$, $\forall p$.
 {\scriptsize \%~${\bf c}_p$ and ${\bf t}_p$ are intermediate variables, and ${\bf c}_p^{\rm{last}}$
 is the channel estimation of the last stage.}
\STATE $\Omega^0\!\!=\!\!\! \ {\Gamma}\!=\! \widetilde{\Gamma}\!=\!\Omega\!=\!\widetilde{\Omega}\!=\!
 \emptyset$; $l_{\text{min}}\!=\!\widetilde l\!=\!0$. {\scriptsize \%~$\Omega^i$ is the estimated
 support set in the $i$th iteration, $\Gamma$, $\widetilde{\Gamma}$, $\Omega$, and $\widetilde{\Omega}$
 are sets, $l_{\text{min}}$ and $\widetilde l$ denote the support indexes.}
\STATE ${\bf b}^0_p \!=\!{\bf r}_p^{[q,G]}\!\!\in \!\!\mathbb{C}^{G\times 1}\!$, $\forall p$.
 {\scriptsize \%~${\bf b}^i_p$ is the residual of the $i$th iteration.}
\STATE $\sum\nolimits_{p = 1}^P {\left\| {{\bf{b}}_p^{{\rm{last}}}} \right\|_2^2}=+\infty$.
 {\scriptsize \%~${{\bf{b}}_p^{{\rm{last}}}}$ is the residual of the last stage.}
\REPEAT
\STATE ${\bf a}_p=\left( {\bf{\Phi}}_p^{[q,G]} \right)^{*} {\bf b}^{i-1}_p$, $\forall p$.
 {\scriptsize \%~Signal proxy is saved in ${\bf a}_p$.}
\STATE $\Gamma \! = \!\arg \max\limits_{\widetilde{\Gamma}} \left\{ \! \sum\nolimits_{p=1}^{P} \!
 \left\| \left( {\bf a}_p \right)_{\widetilde{\Gamma}} \right\|_2^2 ,\left| \widetilde{\Gamma}
 \right| \! = \! {\cal{T}} \! \right\}$. {\scriptsize \%~Identify support.}
\STATE $\left( {\bf t}_p \right)_{ \Omega^{i-1} \cup \Gamma} \!\!=\!\!
 \left(\left( {\bf{\Phi}}_p^{[q,G]} \right)_{ \Omega^{i-1} \cup \Gamma}\right)^{\dag}
\!\!\! {\bf r}_p^{[q,G]}$, $\forall p$. {\scriptsize \%~LS estimation.}
\STATE $\Omega \! = \!\arg \max\limits_{\widetilde{\Omega}} \left\{ \! \sum\nolimits_{p=1}^{P} \!
 \left\| \left( {\bf t}_p \right)_{\widetilde{\Omega}} \right\|_2^2 ,
 \left| \widetilde{\Omega} \right| \! = \! {\cal T} \! \right\}$. {\scriptsize \%~Prune support.}
\STATE $\left( {\bf c}_p \right)_{\Omega} = \left(\left( {\bf{\Phi}}_p^{[q,G]} \right)_{\Omega}\right)^{\dag}
 {\bf r}_p^{[q,G]}$, $\forall p$. {\scriptsize \%~LS estimation.}
\STATE ${\bf b}_p = {\bf r}_p^{[q,G]} - {\bf{\Phi}}_p^{[q,G]} {\bf c}_p$, $\forall p$. {\scriptsize \%~Compute
 the residual.}
\STATE $l_{\rm min}\! = \! \arg \min\limits_{\widetilde l} \left\{ \sum\nolimits_{p=1}^{P}
 \left\| \left[ {\bf c}_p \right]_{\widetilde l} \right\|_2^2 , {\widetilde l} \in \Omega \right\}$.
 {\scriptsize \%~Find the support of the minimum average energy according to the estimated~${\bf c}_p$.}
\IF {$\sum\nolimits_{p=1}^{P} \left\| \left[ {\bf c}_p \right]_{l_{\min}}
 \right\|_2^2 /{P} < p_{\rm{th}}$}
    \STATE {Quit iteration.} {\scriptsize \%~Support index associated with AWGN may be included~in~$\Omega$.}
\ELSIF{$\sum\nolimits_{p=1}^{P} \left\| {\bf b}_p^{\rm{last}} \right\|_2^2  < \sum\nolimits_{p=1}^{P}
 \left\| {\bf b}_p \right\|_2^2$}
    \STATE {Quit iteration.} {\scriptsize \%~Larger~residual~of~the current~stage~than that of the
 last stage indicates that it is unnecessary to~continue~the iteration.}
\ELSIF{$\sum\nolimits_{p=1}^{P} \left\| {\bf b}_p^{i-1} \right\|_2^2 \le \sum\nolimits_{p=1}^{P}
 \left\| {\bf b}_p \right\|_2^2$}
     \STATE $j \!\!= \!\!j\!+\!1$; ${\cal{T}}\!\! = \!\!j$; ${\bf c}_p^{\rm{last}}
            \!\!=\!\! {\bf c}_p$, ${\bf b}_p^{\rm{last}}\! \!=\!\! {\bf b}_p$, $\forall p$.
 {\scriptsize \%~Begin a new stage. The larger residual of current iteration than that of last
 iteration indicates that iteration at current stage~converges.}
\ELSE
    \STATE ${\Omega}^i = {\Omega}$; ${\bf b}_p^i = {\bf b}_p$, $\forall p$; $i  = i+1$.
 {\scriptsize \%~Continue the iteration at the current stage.}
\ENDIF
\UNTIL{$\sum\nolimits_{p=1}^{P} \left\| \left[ {\bf c}_p \right]_{l_{\min}} \right\|_2^2 / {P}
 < p_{\rm{th}}$}
\STATE $\widehat{\bar{\bf h}}_p^{(q)} = {\bf c}_p^{\rm{last}}$, $\forall p$. {\scriptsize \%~Obtain
 the final channel estimation.}
\end{algorithmic}
}
\end{algorithm}

 Compared to the classical SAMP algorithm \cite{SAMP} which recovers one high-dimensional
 sparse signal from single low-dimensional received signal, the proposed DSAMP algorithm
 can simultaneously recover multiple high-dimensional sparse signals with the common
 support set by jointly processing multiple low-dimensional received signals. In terms of
 termination condition, the SAMP algorithm stops the iteration once the residual is smaller
 than a threshold $p_{\rm th}$. By contrast, the proposed DSAMP algorithm has two halting
 criteria. Specifically, if the energy associated with one virtual angular coordinate in
 the estimated support set is smaller than $p_{\rm th}$ or the residual of the current
 stage is larger than that of the previous stage, the algorithm stops. The proposed halting
 criteria ensure the robust signal recovery performance, which will be discussed in
 Section~\ref{convergence} and confirmed by simulations.

 By using the DSAMP algorithm at the BS, we can acquire the estimates of the virtual
 angular domain channels at the pilot subcarriers, i.e., $\widehat{\bar{\bf h}}_p^{(q)}$
 for $1\le p \le P$. Consequently, the actual channel at the $\xi_p$th subcarrier
 dedicated to the $p$th pilot signal can be acquired according to (\ref{equ:channelmode5}),
 yielding
\begin{align}\label{equ:CSI_estimation} % eq11
 \widehat{\bf h}_{\xi_p}^{(q)} =& \left( {\bf A}_{B}^{*} \right)^{\rm T}
 \widehat{\widetilde{\bf h}}_{\xi_p}^{(q)} = \left( {\bf A}_B^{*} \right)^{\rm T}
 \widehat{\bar{\bf h}}_p^{(q)} .
%%\vspace{-6mm}
\end{align}

%%\vspace{-1mm}
\subsection{Closed-Loop Channel Tracking with Adaptive Pilot Design}\label{S3.4}

 Since the channels in $Q$ successive time blocks share the spatially common sparsity,
 in the following $Q-1$ time blocks, we can use the simple LS algorithm to estimate
 the channels at the BS from the feedback pilot signals by utilizing the estimated
 support set $\widehat{\Theta}={\rm{supp}}\big(\widehat{\bar{\bf h}}_p^{(q)}\big)$ and
 the sparsity level $\widehat{S}_a=\big| \widehat{\Theta}\big|$ acquired in the $q$th
 time block. Specifically, for the $q_b$th time block, where $q+1\le q_b\le q+Q-1$, the BS
 first transmits a non-orthogonal pilot to the user, and the user again directly feeds
 back the received pilot signal to the BS. At the BS, similar to (\ref{equ:joint_process2}),
 the feedback pilot signal associated with the $p$th pilot subcarrier ${\bf r}_p^{[q_b,G]}$
 can be expressed as
\begin{align}\label{equ:joint_process7}  % eq 12
  \! \!{\bf r}_p^{[q_b,G]} \! \!= \!\!{\bf S}_p^{[q_b,G]} \!\!\left(  \!{\bf A}_B^{*} \!\right)^{\rm T}\!
 \!\bar{\bf h}_p^{(q_b)}\! \! + \! \!{\bf v}_p^{[q_b,G]}
  \!\!= \! \!{\bf{\Phi}}_p^{[q_b,G]}\bar{\bf h}_p^{(q_b)}  \!\!+ \! \!{\bf v}_p^{[q_b,G]},%%\vspace{-4mm}
\end{align}
 where ${\bf S}_p^{[q_b,G]}$, $\bar{\bf h}_p^{(q_b)}$ and ${\bf v}_p^{[q_b,G]}$ are the
 pilot signal matrix, virtual angular domain channel, and effective noise in the $q_b$th
 time block, respectively. If $\Theta$ and $S_a$ are known, the CSI can be
 acquired using the LS algorithm as
\begin{align}\label{eq13}
 \Big( \widehat{\bar{\bf h}}_p^{(q_b)} \Big)_{\Theta} =& \Big( \Big( {\bf{\Phi}}_p^{[q_b,G]}
 \Big)_{\Theta} \Big)^{\dag} {\bf r}_p^{[q_b,G]},%%\vspace{-2mm}
\end{align}
 which is an unbiased estimator for $\bar{\bf h}_p^{(q_b)}$ that is capable of approaching
 the CRLB \cite{threshold}. The BS can use $\widehat{\Theta}$ and $\widehat{S}_a$, obtained
 in the $q$th time block, to calculate this LS estimate.

 As will be shown in Section~\ref{S4.6}, to acquire the estimate of $\bar{\bf h}_p^{(q_b)}$,
 the required time slot overhead can be reduced to $S_a$. Moreover, the non-orthogonal pilot
 used for channel tracking (for the time blocks of $q+1\le q_b\le q+Q-1$) is very different
 from that used in the $q$th time block. Specifically, to minimize the MSE performance of the
 channel estimation with $G=S_a$, $\big( {\bf{\Phi}}_p^{[q_b,G]} \big)_{\Theta}$ should be a
 unitary matrix scaled by a power factor $\sqrt{G}$. Therefore, for the closed-loop channel
 tracking, we can design the non-orthogonal pilot signal to guarantee this condition, and
 reduce $G$ to $S_a$ while attaining the best MSE performance for the channel estimation.
 Specifically, let $G=S_a$ and ${\bf U}_{S_a}\in \mathbb{C}^{S_a\times S_a}$ be a unitary
 matrix. Then
\begin{align}\label{eq14}
 \Big( {\bf{\Phi}}_p^{[q_b,G]} \Big)_{\Theta} =& \Big( {\bf S}_p^{[q_b,G]} \big(
 {\bf A}_B^{*}\big)^{\rm T} \Big)_{\Theta} = \sqrt{G}{\bf U}_{S_a} ,
\end{align}
 which yields the required non-orthogonal pilot matrix ${\bf S}_p^{[q_b,G]}=\sqrt{G}
 {\bf U}_{S_a}\Big( \Big( \big( {\bf A}_B^{*} \big)^{\rm T} \Big)_{\Theta} \Big)^{\dag}$.

%%\vspace{-1mm}
\section{Performance Analysis}\label{S4}

 The performance analysis of the proposed scheme includes: 1)~the non-orthogonal pilot design
 for the CS based adaptive CSI acquisition; 2)~the theoretical limit of the required time slot
 overhead for the CS based adaptive CSI acquisition; 3)~the placement of pilot subcarriers;
 4)~the computational complexity and convergence of the DSAMP algorithm; 5)~the performance
 bound of the proposed scheme; 6)~the required time slot overhead and the performance analysis
 for the adaptive non-orthogonal pilot based closed-loop channel tracking; and 7)~the selection
 of thresholds for Algorithms 1 and 2.

%%\vspace{-1mm}
\subsection{Non-Orthogonal Pilot Design for CS based Adaptive CSI Acquisition}\label{S4.1}

 In the $q$th time block, the measurement matrices ${\bf{\Phi}}_p^{[q,G]}$ $\forall p$ in
 (\ref{equ:joint_process2}) are very important for guaranteeing the reliable CSI acquisition
 at the BS. Usually, $G \ll M$. Since ${\bf{\Phi}}_p^{[q,G]}={\bf S}_p^{[q,G]}\big(
 {\bf A}_B^{*}\big)^{\rm T}$ and ${\bf A}_B$ is determined by the geometrical structure of
 the antenna array at the BS, the pilot signals ${\bf S}_p^{[q,G]}$ $\forall p$ transmitted
 by the BS should be designed to guarantee the desired robust channel estimation and feedback.

\subsubsection{Restricted Isometry Property (RIP)}

 In CS theory, RIP is used to evaluate the quality of the measurement matrix, in terms of
 the reliable compression and reconstruction of sparse signals. It is proven in \cite{STR_CS}
 that the measurement matrix with its elements following the independent and identically
 distributed (i.i.d.) complex Gaussian distributions satisfies the RIP and enjoys a satisfying
 performance in compressing and recovering sparse signals.

\subsubsection{Processing Multiple Measurement Vectors in Parallel}

 The optimization problem (\ref{equ:target_func}) is essentially different from the
 single-measurement-vector (SMV) and MMV problems in CS.

 The SMV recovers single high-dimensional sparse signal ${\bf f}$ from its low-dimensional
 measurement signal ${\bf d}$, which may be formulated as ${\bf d}={\bf{\Phi}}{\bf f}$, where
 ${\bf{\Phi}}\in \mathbb{C}^{D\times F}$, $D < F$, and the support set $\Xi={\rm{supp}}\{{\bf f}\}$
 with the sparsity level $|\Xi|=S \ll F$.

 On the other hand, the MMV can simultaneously recover multiple high-dimensional sparse
 signals with the common support set and common measurement matrix from multiple
 low-dimensional measurement signals, which may be formulated as ${\bf D} ={\bf{\Phi}}{\bf F}$,
 with ${\bf D}=\left[ {\bf d}_1 ~ {\bf d}_2 \cdots {\bf d}_L\right]$, ${\bf F}=\left[
 {\bf f}_1 ~ {\bf f}_2 \cdots {\bf f}_L\right]$, ${\rm{supp}}\left\{ {\bf f}_1\right\}
 ={\rm{supp}}\left\{ {\bf f}_2\right\}=\cdots ={\rm{supp}}\left\{ {\bf f}_L\right\}=\Xi$, and
 the sparsity level  $\left| \Xi \right| =S$.

 By contrast, our problem (\ref{equ:target_func}) can jointly reconstruct multiple
 high-dimensional sparse signals with the common support set but having different measurement
 matrices, i.e.,
\begin{align}\label{equ:GMMV} % eq15
 {\bf{d}}_l =& {\bf{\Phi}}_l{\bf f}_l, \, 1\le l \le L,
\end{align}
 where ${\bf{\Phi}}_l\in \mathbb{C}^{D\times F}$, $\forall l$. Therefore, our problem can
 be regarded as the GMMV problem, which includes the SMV and MMV problems as
 its special cases. Specifically, if the multiple measurement matrices are identical, our GMMV
 becomes the conventional MMV, and furthermore if $L=1$, it reduces to the conventional SMV.

 Typically, the MMV has the better recovery performance than the SMV, due to the potential
 diversity from multiple sparse signals \cite{STR_CS}. Intuitively, the recovery performance
 of multiple sparse signals with different measurement matrices, as defined in the GMMV,
 should be better than that with the common measurement matrix as given in the MMV. This is
 because the further potential diversity can benefit from different measurement matrices for
 the GMMV. To prove this, we investigate the uniqueness of the solution to the GMMV problem.
 First, we introduce the concept of `spark' and the $\ell_{0}$-minimization based GMMV problem
 associated with (\ref{equ:GMMV}).

\begin{definition}\cite{STR_CS} \label{D1}
 The smallest number of columns of $\bf{\Phi}$ which are linearly dependent is the spark of
 the given matrix $\bf{\Phi}$, denoted by ${\rm spark}(\bf{\Phi})$.
\end{definition}

\begin{problem} \label{P1}
 $\min\limits_{{\bf f}_l,\forall l} \sum\limits_{l=1}^L \left\| {\bf f}_l \right\|_0^2 ,
~{\rm{s.t.}}~{\bf d}_l\!=\!{\bf{\Phi}}_l {\bf f}_l,~{\rm{supp}}\left\{ {\bf f}_l\right\}\!=\! \Xi,~\forall l$. %1 \le p \le P$.
\end{problem}

 For the above $\ell_{0}$-minimization based GMMV problem, the following result can be obtained.

\begin{theorem}\label{T1}
 For ${\bf{\Phi}}_l$, $1\le l\le L$, whose elements obey an i.i.d. continuous distribution,
 there exist full rank matrices ${\bf{\Psi}}_l$ for $2\le l\le L$ satisfying
 $\left( {\bf{\Phi}}_l \right)_{\Xi}={\bf{\Psi}}_l\left( {\bf{\Phi}}_1 \right)_{\Xi}$
 if we select $\left( {\bf{\Phi}}_1\right)_{\Xi}$ as the bridge, where $\Xi$ is the
 common support set. Consequently, ${\bf f}_l$ for $1\le l\le L$  will be the unique solution
 to Problem~\ref{P1} if
\begin{align}\label{equ:GMMV3} % eq16
 2S <\rm{spark}\left( {\bf{\Phi}}_1 \right) - 1 + \rm{rank}\big\{ \widetilde{\bf D}\big\} ,
\end{align}
 where $\widetilde{\bf D}=\left[ {\bf d}_1 ~ {\bf{\Psi}}_2^{-1} {\bf d}_2 \cdots
 {\bf{\Psi}}_L^{-1} {\bf d}_L\right]$.
\end{theorem}

\begin{proof}
 Consider (\ref{equ:GMMV}) with ${\bf{\Phi}}_l\in \mathbb{C}^{D\times F}$ for$1\le l\le L$,
 whose elements follow an i.i.d. continuous distribution. The common support set is
 $\Xi ={\rm{supp}}\left\{ {\bf f}_l \right\}$ with the sparsity level $\left| \Xi  \right|
 =S$. This GMMV can be expressed as
\begin{align}\label{aA1} % eq17
 {\bf d}_l =& \left( {\bf{\Phi}}_l \right)_{\Xi} \left( {\bf f}_l \right)_{\Xi}
 = {\bf Z}_l \left( {\bf f}_l \right)_{\Xi} , \, 1\le l\le L.
\end{align}
 The random matrix ${\bf Z}_l=\left( {\bf{\Phi}}_l \right)_{\Xi}\in \mathbb{C}^{D\times S}$
 is a tall matrix, as $D > S$. Clearly, ${\rm{rank}}\left\{ {\bf Z}_l \right\}=S$ with high
 probability, since the measure of the set $\left\{ {\bf Z}_l\in \Omega_{\bf Z}:
 {\rm{rank}}\{{\bf Z}_l\} < S\right\}$ is zero~\cite{cedu}. If we take
 $\left({\bf{\Phi}}_1\right)_{\Xi}$ as the bridge, then there exist the full rank matrices
 ${\bf{\Psi}}_l$, $2\le l\le L$, satisfying $\left({\bf{\Phi}}_l\right)_{\Xi}= {\bf{\Psi}}_l
 \left({\bf{\Phi}}_1\right)_{\Xi}$, and thus
\begin{align} % eq18
 {\bf{\Psi}}_l^{-1}{\bf d}_l =& \left( {\bf{\Phi}}_1 \right)_{\Xi}
 \left( {\bf f}_l \right)_{\Xi} = {{\bf{\Phi }}_1}{{\bf{f}}_l}.
\end{align}
 In this way, the GMMV is converted to the `equivalent' MMV
\begin{align}\label{aA3} % eq19
 \widetilde{\bf D} =& {\bf{\Phi}}_1 {\bf F} ,
\end{align}
 where ${\bf F}=\left[ {\bf f}_1 ~ {\bf f}_2 \cdots {\bf f}_L \right]$. Applying the existing
 result for the MMV given in \cite{MMV-J}, (\ref{equ:GMMV3}) can be directly obtained.
\end{proof}

 From Theorem~\ref{T1}, it is clear that the achievable diversity gain introduced by
 diversifying measurement matrices and sparse vectors is determined by $\rm{rank}\big\{
 \widetilde{\bf D}\big\}$. The larger $\rm{rank}\big\{ \widetilde{\bf D}\big\}$ is, the
 more reliable recovery of sparse signals can be achieved. Hence, compared to the SMV and
 MMV, more reliable recovery performance can be achieved by the proposed GMMV. For the
 special case that multiple sparse signals are identical, the MMV reduces to the SMV
 since ${\rm rank}\left({\bf D}\right)=1$, and there is no diversity gain by introducing
 multiple identical sparse signals. However, the GMMV in this case can still achieve
 diversity gain which comes from diversifying measurement matrices.

\subsubsection{Pilot Design for CS Based Adaptive CSI Acquisition}

 According to the above discussions, a measurement matrix whose elements follow an i.i.d.
 Gaussian distribution satisfies the RIP. Furthermore, diversifying measurement matrices
 can further improve the recovery performance of sparse signals. This  enlightens us to
 appropriately design pilot signals.

 Specifically, each element of pilot signals is given by
\begin{align}\label{equ:pilot_element1} % eq 20
 \big[ {\bf S}_p^{[q,G]}\big]_{t,m} =& e^{\textsf{j}\theta_{t,m,p}} , \, 1\le t \le G,
 \, 1\le m \le M ,
\end{align}
 where ${\bf S}_p^{[q,G]}\in\mathbb{C}^{G\times M}$, and each $\theta_{t,m,p}$ has the
 i.i.d. uniform distribution in $[0, ~ 2\pi)$, namely, the i.i.d.
 ${\cal{U}}\left[ 0, ~ 2\pi\right)$. Note that the pilot signals for the CS based adaptive
 CSI acquisition are fixed once they have been designed. Moreover, when designing the pilot
 signals, the worst case of $G=M$ has to be considered. It is readily seen that
 the designed pilot signals (\ref{equ:pilot_element1}) guarantee that the elements of
 ${\bf{\Phi}}_p^{[q,G]}$ of (\ref{equ:target_func}), $\forall p$,  obey the i.i.d. complex
 Gaussian distribution with zero mean and unit variance, i.e., the i.i.d. ${\cal{CN}}(0,1)$.
 Hence, the proposed pilot signal design is `optimal', in terms of the reliable compression
 and recovery of sparse angular domain channels.

%%\vspace{-1mm}
\subsection{Time Slot Overhead for CS based Adaptive CSI Acquisition}\label{S4.2}

 According to Theorem~\ref{T1}, for the optimization problem (\ref{equ:target_func}),
 $\widetilde{\bf D}\!=\!{\bf{\Phi}}_1^{[q,G]} {\bf F}$ with
 $\widetilde{\bf D}\!=\!\big[ {\bf r}_1^{[q,G]} ~ {\bf{\Psi}}_2^{-1} {\bf r}_2^{[q,G]}$ $\cdots
 {\bf{\Psi}}_{P}^{-1} {\bf r}_{P}^{[q,G]}\big]$ and ${\bf F}\!\!=\!\!\left[ \bar{\bf h}_1^{(q)} ~
 \bar{\bf h}_2^{(q)} \cdots \bar{\bf h}_{P}^{(q)} \!\right]$. Since $\big| {\rm supp}\big\{\bar{\bf h}_p^{(q)}\!
 \big\}\!\big|\!\!\!=\!\!S_a$, it is clear that
\begin{align}\label{equ:T_overhead1} % eq 21
 {\rm{rank}}\big\{ \widetilde{\bf D} \big\} \le {\rm{rank}}\left\{ {\bf F} \right\}\le S_a .
\end{align}
 Moreover, as ${\bf{\Phi}}_1^{[q,G]}\in \mathbb{C}^{G\times M}$,
\begin{align}\label{equ:T_overhead2}  % eq 22
 {\rm{spark}}\big( {\bf{\Phi}}_1^{[q,G]} \big) \in \left\{2,3,\cdots ,G+1\right\} .
\end{align}
 Substituting (\ref{equ:T_overhead1}) and (\ref{equ:T_overhead2}) into (\ref{equ:GMMV3})
 yields $G\ge S_a+1$. Therefore, the smallest required time slot overhead is $G=S_a+1$.
 As discussed in Section~\ref{CSI acquisition}, an appropriate value of $G$ that
 ensures the reliable CSI acquisition is adaptively determined by Algorithm~\ref{alg1}.
 By increasing the number of measurement vectors $P$, the required time slot overhead
 for reliable channel estimation can be reduced, since more measurement matrices
 and sparse signals can increase ${\rm{rank}}\big\{  \widetilde{\bf D}\big\}$.

%%\vspace{-1mm}
\subsection{Frequency-Domain Placement of Pilot Signals}\label{S4.3}

 Like any OFDM channel estimator, the proposed adaptive channel estimation and feedback
 scheme only estimates the channels at pilot subcarriers. Channels at data subcarriers
 are usually acquired based on the estimated channels at pilot subcarriers by using the
 off-the-shelf interpolation algorithms \cite{matlab}.  Clearly the frequency-domain
 placement of pilot signals $\Omega_{\xi}$ significantly influences the achievable
 performance of an interpolation algorithm. Additionally, due to the frequency-domain
 correlation of wireless channels, the channels of adjacent subcarriers exhibit strong
 correlation. Hence, two adjacent subcarriers both dedicated to the pilot may result in
 $\widetilde{\bf D}$ to be rank deficient. We adopt the widely used uniformly-spacing pilot
 placement with the spacing equal to the coherence bandwidth \cite{matlab}, which can
 reduce the correlation between different virtual angular domain channels, so that more
 diversity gain from the multiple sparse channels can be achieved.

%%\vspace{-1mm}
\subsection{Performance Analysis of Proposed DSAMP Algorithm}\label{S4.4}

\subsubsection{Complexity}

 The computational complexity of the proposed DSAMP algorithm (Algorithm~\ref{alg:Framwork})
 in each iteration mainly depends on the following operations.

 {\bf Signal proxy} (\emph{line 7}):  The matrix-vector multiplication involved has the
 complexity on the order of $\textsf{O}\big(P M G\big)$.

 {\bf $\ell_{2}$-norm operation} (\emph{lines 8, 10, 13, 14, 16, 18}, and \emph{23}): The
 computational complexity is $\textsf{O}\big(P\big)$.

 {\bf Identifying or Pruning} (\emph{lines 8} and \emph{10}): The cost to locate the
 largest $\cal{T}$ entries of a size-$M$ vector is $\textsf{O}(M)$ \cite{STR_CS}.

 {\bf LS operation} (\emph{lines 9} and \emph{11}): LS solution has the computational
 complexity on the order of $\textsf{O}\big(P(2G{\cal T}^2+{\cal T}^3)\big)$ due to
 the joint recovery of $P$ sparse signals \cite{LS}.

 {\bf Residual computation} (\emph{line 12}): The complexity of computing the residual is
 $\textsf{O}\big(P M G\big)$.

\begin{table}[tp!]
%%\vspace{-1mm}
\caption{Computational Complexity to Estimate One Sparse Signal}
\label{TAB1}
%%\vspace{-9mm}
\begin{center}
\begin{tabular}{c |c }
\hline
\hline
 {Algorithm} & \multicolumn{1}{c}{Number of complex multiplications in each iteration} \\
 \hline
 OMP   & $2GM + M + 2G{i^2} + {i^3}$ \\ % ($i$ is the index of iterations)
 SP    & $2GM + G+M + 2{S_a} + 2(2GS_a^2 + S_a^3)$ \\
 SAMP  & $2GM + {{G}} + M + 3{S_a} + 2(2G{j^2} + {j^3})$ \\% ($j$ is the adaptive sparsity level)
 DSAMP & $2GM + {{G}} + M + 3{S_a} + 2(2G{j^2} + {j^3})$ \\
\hline
\hline
\multicolumn{2}{c}{Note: $i$ denotes the iteration index, and $j$ denotes the stage index.}
\end{tabular}
\end{center}
%%\vspace{-12mm}
\end{table}

 Obviously the matrix inversion implemented in Algorithm~\ref{alg:Framwork} for LS
 operation contributes to most of the computational complexity. Table~\ref{TAB1} compares
 the complexity of the proposed DSAMP algorithm, classical OMP algorithm, SP algorithm
 \cite{STR_CS}, and SAMP algorithm, in terms of the number of required complex
 multiplications in each iteration to estimate one sparse signal. It is clear that the four
 algorithms have the same order of computational complexity.

\subsubsection{Stopping Criteria}\label{convergence}

 For the conventional SAMP algorithm, the iterative procedure stops when the residual
 is less than a given threshold. By contrast, the proposed DSAMP algorithm has two
 halting criteria, and meeting either of them will trigger the termination of the
 iterative procedure. Regarding the first halting criterion, when the average energy of
 the wireless channels at a certain virtual angular coordinate is lower than the noise
 floor (\emph{lines 14} and \emph{23}), the iterative procedure stops. When the residual
 of the current stage becomes larger than that of the previous stage (\emph{line 16}),
 the second halting criterion is met and the algorithm also terminates.

 Due to $S_a \ll M$, after coordinates accounting for the majority of the channel energy
 is achieved, the next iteration will include a virtual angular coordinate that is
 dominated by the AWGN. The energy of such a new coordinate is usually lower than the
 noise floor. The first stopping criterion is designed to detect this situation and to
 terminate the algorithm when an appropriate number of virtual angular domain coordinates
 have been tracked.

 As for the second halting criterion, the DSAMP algorithm is similar to the conventional
 SP algorithm in each stage with the fixed sparsity level, which can guarantee the sparse
 signal recovery with the exact sparsity level. The residual of the stage with the exact
 sparsity level is usually smaller than that with the incorrect sparsity level. Therefore,
 the DSAMP algorithm stops at the stage when the smallest residual is reached, which tends
 to be the stage associated with the exact sparsity level of the channels in the virtual
 angular domain.

%%\vspace{-1mm}
\subsection{Performance Bound of Channel Estimation}\label{S4.5}

 By omitting $q$, $p$, and $\xi_p$ in (\ref{equ:CSI_estimation}) for simplicity, the
 variance of the channel estimation can be expressed as
\begin{align}\label{eq23} % eq23
 {\rm var}\left\{\widehat{\bf h}\right\} & = {\rm E}\left\{\left\|\widehat{\bf h} - {\bf h}\right\|_2^2
 \right\} = {\rm E}\left\{\left\|\left( {\bf A}_B^{*} \right)^{\rm T} \widehat{\bar{\bf h}} -
 \left( {\bf A}_B^{*} \right)^{\rm T} \bar{\bf h} \right\|_2^2 \right\}  \nonumber \\ &
 = {\rm E}\left\{\left\|\widehat{\bar{\bf h}} - \bar{\bf h} \right\|_2^2 \right\}
 = {\rm{var}}\left\{ \widehat{\bar{\bf h}}\right\} .
\end{align}
 Consider the CRLB for the estimation problem associated with (\ref{equ:joint_process2})
 given the true channel $\bar{\bf h}$ and the support set $\Theta$. Again for notational
 simplicity,  $q$, $G$, and $p$ in  ${\bf r}_p^{[q,G]}$, ${\bf{\Phi}}_p^{[q,G]}$, and
 ${\bf v}_p^{[q,G]}$ are omitted. Since the distribution of ${\bf v}$ is
 ${\cal CN}\left( {\bf 0},\sigma^2{\bf I}_G \right)$, the conditional probability density
 function (PDF) of ${\bf r}$ given $\bar{\bf h}$ is
\begin{align}\label{equ:pdf} % eq24
 p_{{\bf r} | \bar{\bf h}} \big( {\bf r} | \bar{\bf h} \big) =& \frac{1}{\left( \pi \sigma^2 \right)^{G}}
 e^{ - \frac{\left\| {\bf r} - ( {\bf{\Phi}} )_{\Theta}
 ( \bar{\bf h} )_{\Theta} \right\|_2^2}{\sigma^2} } ,
\end{align}
 where $\sigma^2$ is the power of the effective noise. The element at the $s_i$th-row and
 $s_j$th-column of the Fisher information matrix ${\cal I}\left( \big(\bar{\bf h}\big)_{\Theta}
 \right)$ associated with this estimation problem is
\begin{align}\label{equ:fisher} % eq25
 \left[ {\cal{ {I}}} \left( \big(\bar{\bf h}\big)_{\Theta} \right) \right]_{s_i,s_j}
 =& \frac{1}{\sigma^2}\left[ \left(\left( {\bf{\Phi}} \right)_{\Theta} \right)^{*} \left(
 {\bf{\Phi}}\right)_{\Theta} \right]_{s_i,s_j},%%\vspace{-2mm}
\end{align}
 where $1 \le s_i,s_j \le |\Theta|$. Therefore, we have
\begin{align}\label{equ:CRLB} % eq26
\!\! {\rm var}\big\{ \widehat{\bar{\bf h}} \big\}\! \ge & {\rm Tr}\left\{{\big( \bf{\cal{I}}
\! \big( \big(\bar{\bf h}\big)_{\!\Theta}\! \big)\! \big)}^{-1}\!\! \right\} \!\!  = \!
 \sigma^2 {\rm Tr}\left\{\!\big(\! \big(\! \big( \!{\bf{\Phi}} \big)_{\!\Theta}\! \big)^{\!*} \! \big(
 {\bf{\Phi}} \big)_{\!\Theta} \big)^{\!-1\!} \! \right\}.%%\vspace{-2mm}
\end{align}
 Let $\lambda_1,\lambda_2,\cdots ,\lambda_{S_a}$ be the $S_a$ eigenvalues of the matrix
 $\left(\left( {\bf{\Phi}} \right)_{\Theta}\right)^{*}\left({\bf{\Phi}}\right)_{\Theta}
 \in \mathbb{C}^{S_a\times S_a}$. It is clear that
\begin{align}\label{equ:qiyizhi} % eq27
 {\rm Tr}\left\{ \left(\left(\left( {\bf{\Phi}} \right)_{\Theta}\right)^{*}
 \left( {\bf{\Phi}}\right)_{\Theta} \right)^{-1} \right\} =& \sum\nolimits_{i=1}^{S_a}
 \lambda_i^{-1} ,
\end{align}
 which can be calculated after the pilot signals, the geometrical structure of the BS
 antenna array, and the support set of the channel vectors in the virtual angular domain
 are given.

 However, the support set $\Theta$ is `random' since the channel vectors in practice are
 random and the elements of the measurement matrix ${\bf{\Phi }}$ obey the i.i.d.
 ${\cal{CN}}(0,1)$. Thus we should consider the `expectation' of the CRLB defined by
\begin{align}\label{eq28} % eq28
 {\rm E}\left\{ {\rm var} \left\{ \widehat{\bar{\bf h}} \right\} \right\} \ge {\rm E}
 \left\{ \sigma^2\sum\nolimits_{i=1}^{S_a} \lambda_i^{-1} \right\} .
\end{align}
 For the matrix $\left(\left( {\bf{\Phi}} \right)_{\Theta}\right)^{*}\left( {\bf{\Phi}}
 \right)_{\Theta}$ with the elements of ${\bf{\Phi}}$ obeying the i.i.d. ${\cal{CN}}(0,1)$,
 its eigenvalues $\left\{ \lambda_i \right\}_{i=1}^{S_a}$ obey the following joint
 distribution \cite{qiyizhi}
\begin{align}\label{eq29}
\!\!\!\!  p_{\tilde \lambda}\!\left( \lambda_1,\!\lambda_2,\!\cdots \!,\!\lambda_{S_a} \! \right)
 \!= \! e^{-\!\!\sum\limits_{i=1}^{S_a} \! \lambda_i}\!\! \prod\limits_{i=1}^{S_a}\!\!\left(\!\!
 \frac{\lambda_i^{G-S_a}}{\left( S_a - i \right)! \, i !}\!\! \prod\limits_{j > i}^{S_a}
 \!\!\left( \lambda_j \!- \!\lambda_i \!\right)^2\!\!\right)\!.
\end{align}
 Consequently, the expectation of the CRLB can be written as
\begin{align}\label{equ:jifenjie} % eq 30
\!\!\! {\rm E}\! \left\{ \!{\rm var} \!\left\{ \widehat{\bar{\bf h}}\! \right\}\! \right\}\! \ge \!\!\!
 \int\limits_0^{\infty} \!\! \cdots \!\!\int\limits_0^{\infty}\!\! \sigma^2\!\!\sum\limits_{i=1}^{S_a}\!\! \lambda_i^{-1}
 p_{\tilde \lambda}\!\!\left( \lambda_1, \!\cdots \! ,\! \lambda_{S_a} \!\right)\! d{\lambda_1} \!\cdots \! d{\lambda_{S_a}}.
\end{align}
 Since the computation of (\ref{equ:jifenjie}) can be highly complex, in practice we adopt the
 performance of the oracle LS estimator as the performance bound in the simulation study.

%%\vspace{-1mm}
\subsection{Adaptive Pilot Design and Required Time Slot Overhead for Closed-Loop Channel Tracking}\label{S4.6}

 For the simplicity of analysis, the true support set
 $\Theta$ and the sparsity level $S_a$ of the virtual angular domain channels are assumed to
 have been acquired by the CS based adaptive CSI acquisition. Clearly, if $S_a$ is known, the
 smallest time slot overhead for CSI acquisition can be reduced to $G=S_a$.

\begin{figure}[!tp]
%%\vspace{-1mm}
\begin{center}
\hspace*{2mm}\includegraphics[width=0.95\columnwidth,keepaspectratio]{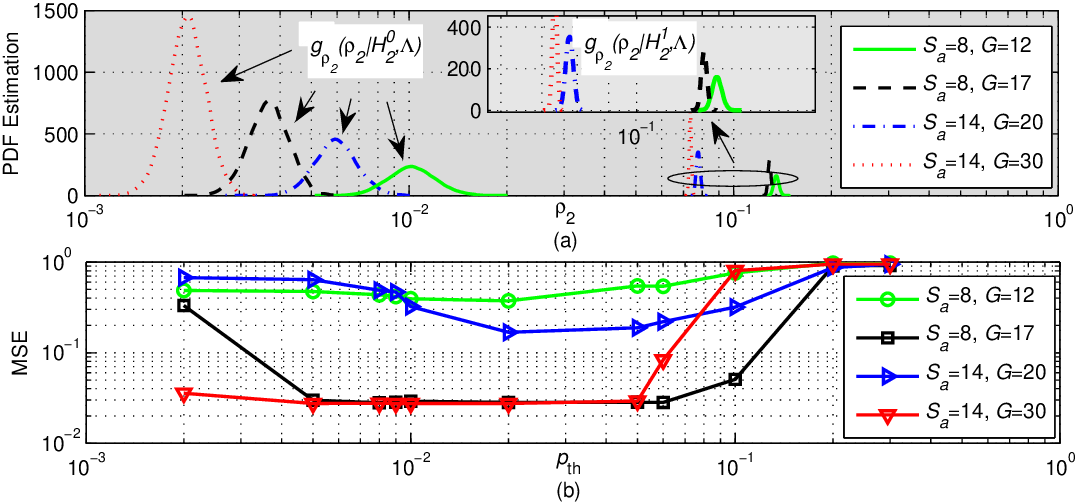}
\end{center}
%%\vspace{-10mm}
\caption{The selection of threshold $p_{\rm th}$ in Algorithm~2 given
 $\mbox{SNR}=15$\,dB and $P=64$: (a)~estimated PDFs of $g_{\rho_2}\left( \rho_2 |
 {\cal H}_2^0,\Lambda \right)$ and $g_{\rho_2}\left( \rho_2 | {\cal H}_2^1,\Lambda \right)$;
 and (b)~MSE performance of the DSAMP algorithm as the function of $p_{\rm th}$.}
\label{fig:A2} % Fig 4
%%\vspace{-10mm}
\end{figure}

 With the known $\Theta$, by exploiting the arithmetic-harmonic means inequality \cite{LS},
 (\ref{equ:qiyizhi}) can be further expressed as
\begin{align}\label{equ:crlb_dai1} % eq 31
\!\!\! {\rm{Tr}}\!\left\{\! \left( \left( \left( {\bf{\Phi}} \right)_{\Theta}\right)^{\!*}\!\left( {\bf{\Phi}}
 \right)_{\Theta}\! \right)^{\!-1} \!\right\} \!\ge\! \frac{S_a^2}{\sum\limits_{i=1}^{S_a}\!\! \lambda_i} \!=\!
 \frac{S_a^2}{{\rm{Tr}}\left\{\! \left(\! \left( {\bf{\Phi}} \right)_{\Theta}\right)^{*}\!
 \left( {\bf{\Phi}}\!\right)_{\Theta}\!\right\} },
\end{align}
 where the equality holds if and only if $\lambda_1=\lambda_2= \cdots =\lambda_{S_a}$. This
 indicates that $\left(\left({\bf{\Phi}}\right)_{\Theta}\right)^*\left({\bf{\Phi}}\right)_{\Theta}$
 should be a diagonal matrix with identical diagonal elements to approach the lower bound. In
 particular, for $\big({\bf{\Phi}}\big)_{\Theta}$ with ${\rm{Tr}}\left\{ \left( \left(
 {\bf{\Phi}} \right)_{\Theta}\right)^{*}\left( {\bf{\Phi}}\right)_{\Theta} \right\}=S_a G$,
\begin{align}\label{equ:crlb_dai2}  % eq 32
 {\rm{var}} \left\{ \widehat{\bar{\bf h}} \right\} \ge {\sigma^2 S_a}/{G} ,
\end{align}
 and the lower bound of (\ref{equ:crlb_dai2}) is attained if $\big({\bf{\Phi}}\big)_{\Theta}$
 is a unitary matrix scaled by the factor $\sqrt{G}$. This has inspired us to design the pilot
 signal matrix as ${\bf S}=\sqrt{G} {\bf U}_{S_a} \left(\left(\left( {\bf A}_B^{*}\right)^{\rm T}
 \right)_{\Theta}\right)^{\dag}$, where ${\bf U}_{S_a}\in \mathbb{C}^{S_a \times S_a}$ is a
 unitary matrix. With this non-orthogonal pilot matrix, the lower bound of (\ref{equ:crlb_dai2})
 is attained, i.e., ${\rm{var}}\left\{ \widehat{\bar{\bf h}} \right\} =\sigma^2 S_a/G=\sigma^2$.

%%\vspace{-1mm}
\subsection{Selection of Thresholds for Algorithms 1 and 2}\label{S4.7}

\subsubsection{$p_{\rm{th}}$ in Algorithm~\ref{alg:Framwork}}\label{S4.7.1}

 Consider the case $\Lambda$ that for the stage of ${\cal T}=S_a + 1$ in
 Algorithm~\ref{alg:Framwork}, the final estimated support set, denoted as $\widehat{\Omega}_{\cal T}$,
 is the proper superset of the true support set $\Theta$, i.e., $\widehat{\Omega}_{\cal T}
 \supsetneqq \Theta$. This case implies that $\widehat{\Omega}_{\cal T}$ includes the
 support index associated with the noise. Define the test statistic as
 $\rho_2=\sum\nolimits_{p=1}^{P}\left\| \left[ {\bf c}_p \right]_{\widetilde{l}}
 \right\|_2^2 \big/ P$ with $\widetilde{l}\in \widehat{\Omega}_{\cal T}$ (\emph{lines 13}
 and \emph{14} in Algorithm~\ref{alg:Framwork}). Two complete hypotheses for the case
 $\Lambda$ are defined as: ${\cal H}_2^0$, indicating $\widetilde{l}\in \Theta$, and
 ${\cal H}_2^1$, indicating $\widetilde{l}\notin \Theta$. Furthermore, denote the PDFs
 of $\rho_2$ under ${\cal H}_2^0$ and ${\cal H}_2^1$ as $g_{\rho_2}\left( \rho_2 |
 {\cal H}_2^0,\Lambda \right)$ and $g_{\rho_2}\left( \rho_2 | {\cal H}_2^1,\Lambda \right)$,
 respectively. By using the \emph{ksdensity} function of MATLAB, we can obtain the estimated
 PDFs according to Monte-Carlo simulations, since the closed-form expressions are difficult
 to obtain.

\begin{figure}[!tp]
%%\vspace{-6mm}
\begin{center}
\includegraphics[width=0.95\columnwidth, keepaspectratio]{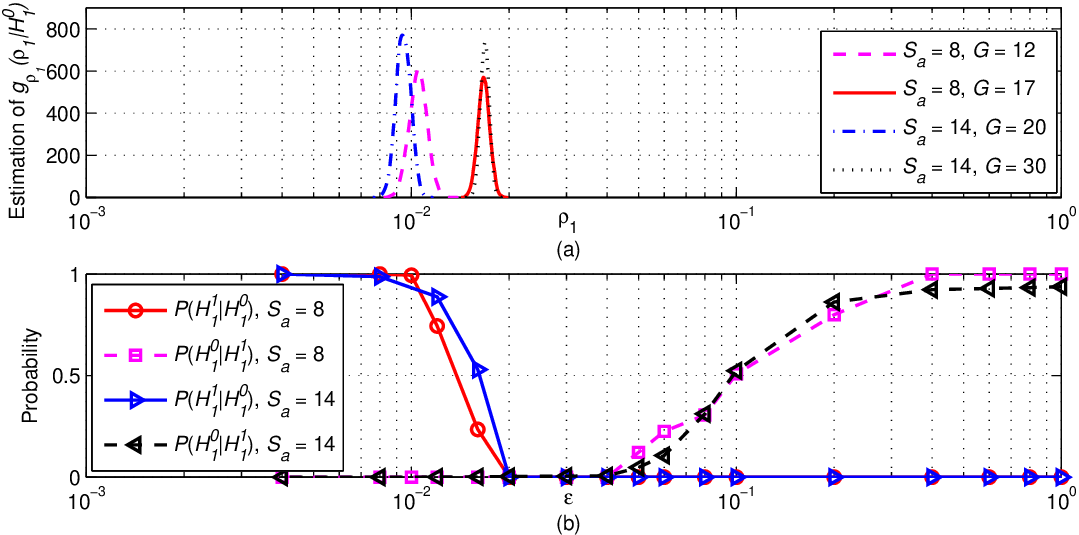}
\end{center}
%%\vspace{-10mm}
\caption{The selection of threshold $\varepsilon$ in Algorithm~\ref{alg1}
 given $\mbox{SNR}=15$\,dB, $G_0=10$, and $P=64$: (a)~estimated PDF of
 $g_{\rho_1}\left( \rho _1 | {\cal{H}}_1^0 \right)$; and (b)~$\Pr \left( {\cal H}_1^1|{\cal H}_1^0 \right)$
 and $\Pr \left( {\cal H}_1^0|{\cal H}_1^1 \right)$ as the functions of $\varepsilon$.}
\label{fig:A1} % Fig 5
%%\vspace{-1mm}
\end{figure}

 Fig.~\ref{fig:A2}\,(a) depicts the estimated PDFs of $g_{\rho_2}\left( \rho_2 |
 {\cal H}_2^0,\Lambda \right)$ and $g_{\rho_2}\left( \rho_2 | {\cal H}_2^1,\Lambda \right)$
 with typical values of $S_a$ and $G$, given $\mbox{SNR}=15$\,dB and $P=64$.
 Fig.~\ref{fig:A2}\,(b) provides the MSE performance of Algorithm~\ref{alg:Framwork} as the
 function of $p_{\rm{th}}$, which indicates that $p_{\rm{th}}=0.02$ achieves good MSE
 performance given typical values of $S_a$ and $G$. Following a similar procedure,
 suitable values of $p_{\rm{th}}$ for different SNRs can be obtained.

\subsubsection{$\varepsilon$ in Algorithm~\ref{alg1}}\label{S4.7.2}

 Consider the test statistic $\rho_1=\sum\nolimits_{p=1}^P\Big\| {\bf r}_p^{[q,G]}-{\bf \Phi}_p^{[q,G]}
 \widehat{\bar{\bf h}}_p^{(q)} \Big\|_2^2\big/ (GP)$ (\emph{line 6} in Algorithm~\ref{alg1}
 with iteration index $i$ omitted) and the two complete hypotheses ${\cal{H}}_1^0$
 and ${\cal{H}}_1^1$, where ${\cal{H}}_1^0$ indicates that the support set of
 $\{\widehat{\bar{\bf h}}_p^{(q)}\}_{p=1}^{P}$ is correct, and ${\cal{H}}_1^1$ is
 complementary to ${\cal{H}}_1^0$. Under ${\cal{H}}_1^0$, $\rho_1=\sum\nolimits_{p=1}^P
 \Big\| \big( {\bf{I}}-{\bf{\Phi}}\big)_{\Theta} \big( \big({\bf{\Phi}}\big)_{\Theta} \big)^{\dag}
 \big){\bf{v}}_p^{[q,{G}]} \Big\|_2^2 \big/ (G P)$. However, under ${\cal{H}}_1^1$, the
 closed-form expression of ${\rho _1}$ is difficult to derive. Similar to Fig.~\ref{fig:A2},
 Fig.~\ref{fig:A1}\,(a) provides the estimates of the PDF $g_{\rho_1}\left( \rho_1 |{\cal{H}}_1^0 \right)$
 with typical values of $S_a$ and $G$, given $\mbox{SNR}=15$\,dB, $G_0=10$ and $P=64$. According
 to Neyman-Pearson criterion \cite{threshold}, an appropriate threshold $\varepsilon$ should
 minimize the probability of false alarm given the probability of miss. Fig.~\ref{fig:A1}\,(b)
 depicts the simulated probability of false alarm $\Pr \left( {{\cal{H}}_1^1|{\cal{H}}_1^0} \right)$
 and the miss probability $\Pr \left( {{\cal{H}}_1^0|{\cal{H}}_1^1} \right)$ of Algorithm~\ref{alg1}
 as the functions of $\varepsilon$, where $p_{\rm{th}}=0.02$ is used in the simulation.
 The results of Fig.~\ref{fig:A1}\,(b) indicate that $\varepsilon =0.03$ minimizes both
 $\Pr \left( {\cal H}_1^1|{\cal H}_1^0 \right)$ and $\Pr \left( {\cal H}_1^0|{\cal H}_1^1 \right)$
 given typical values of $S_a$ and $G$. Similarly, appropriate values of $\varepsilon$ for
 different SNRs can be obtained.

%%\vspace{-1mm}
\section{Simulation Results}\label{S5}

\begin{figure}[!tp]
%%\vspace{-1mm}
\begin{center}
\includegraphics[width=1.1\columnwidth, keepaspectratio]{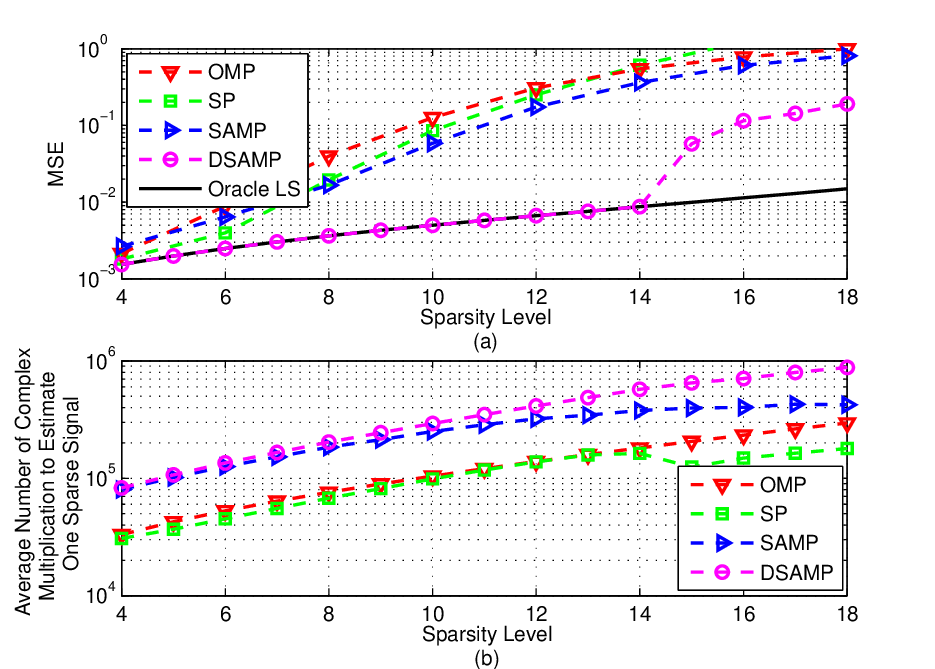}
\end{center}
%%\vspace{-11mm}
\caption{Performance comparison of different CS algorithms as functions of
 the sparsity level $S_a$ given $G=30$, $P=64$ and $\mbox{SNR}=20$\,dB: (a)~MSE
 performance, and (b)~computational complexity.}
\label{fig:MSE_algorithm} % Fig 6
%%\vspace{-9mm}
\end{figure}

 Massive MIMO system with the ULA of $M=128$ antennas and $d ={\lambda}/{2}$ was considered.
 The spatial angle spread varied from $10^\circ$ to $20^\circ$ \cite{{Haifan},{JSDM}}, and
 thus the effective sparsity level in the virtual angular domain $S_a$ was in the range of 8
 to 14. In the simulations, $f_c=2$\,GHz, $B_s=10$\,MHz, $N=2048$, and $v= 36$\,km/h, while
 the channels in the virtual angular domain exhibited the spatially common sparsity over
 $Q=5$ time blocks. The length of the guard interval was 64, which indicates that the system
 can combat the maximum delay spread of $6.4\,\mu$s \cite{channel_model_for4g}, and thus we
 adopted $P=64$ \cite{orthogonal}. The threshold parameters, $\varepsilon$ in Algorithm~\ref{alg1}
 and $p_{\rm{th}}$ in Algorithm~\ref{alg:Framwork}, were selected according to Section~\ref{S4.7}.
 Specifically, we set $p_{\rm{th}}$ to 0.06, 0.02, 0.01, 0.008, and 0.005, while $\varepsilon$
 to 0.08, 0.03, 0.0.09, 0.003, and 0.001, respectively, at the SNR of 10\,dB, 15\,dB, 20\,dB,
 25\,dB and $\ge 30$\,dB.  The oracle LS estimator and the CRLB were used as the benchmarks
 for the CS based adaptive CSI acquisition and the following closed-loop channel tracking,
 respectively. The time slot overhead $G$ employed in the closed-loop channel tracking scheme
 was set to the estimated sparsity level obtained by the CS based adaptive CSI acquisition stage.
 The joint OMP (J-OMP) based CSI acquisition scheme \cite{Rao1} was also adopted for comparison.

 Fig.~\ref{fig:MSE_algorithm} compares the MSE performance and complexity of four CS
 algorithms under various sparsity levels $S_a$. In the simulations, $P\!=\!64$ sparse
 signals with the length of $M\!=\!128$ had the common sparsity, the measurement dimension
 was $G\!=\!30$, and $\mbox{SNR}\!=\!20$\,dB, while $P$ measurement matrices were mutually
 independent with elements obeying the i.i.d ${\cal{CN}}(0,1)$. Note that the conventional
 OMP and SP algorithms require $S_a$ as the priori information. Fig.~\ref{fig:MSE_algorithm}\,(a)
 shows that the DSAMP algorithm achieves the best MSE performance and it approaches the
 oracle LS estimator for $S_a\le 14$\footnote{The DSAMP algorithm suffers
 from certain performance loss, compared to the oracle LS estimator in the noisy scenario
 with $G\le 2S_a$. For $G=2S_a$,  the case of $\Omega^{i-1}\cap\Gamma =\emptyset$
 ({\it line 9}) and $\Omega =\Omega^{i-1}$ ({\it line 13}) may repeatedly appear due to
 noise, resulting in the failure of the backtracking function of {\it lines 7$\sim$12}.
 For $G < 2S_a$, the case of $\Omega^{i-1} \cap \Gamma =\emptyset$ can lead to a poor LS
 estimation ({\it line 9}) due to $|\Omega|>G$. Also see \cite{SAMP}.}. This is because
 the DSAMP algorithm jointly estimates $P$ sparse signals by exploiting the common sparsity.
 Moreover, Fig.~\ref{fig:MSE_algorithm} (b) shows that the complexity of the DSAMP algorithm
 is slightly higher than those of its counterparts, but all the four CS algorithms have the
 same order of complexity.

\begin{figure}[!tp]
%%\vspace{-6mm}
\begin{center}
\includegraphics[width=0.95\columnwidth, keepaspectratio]{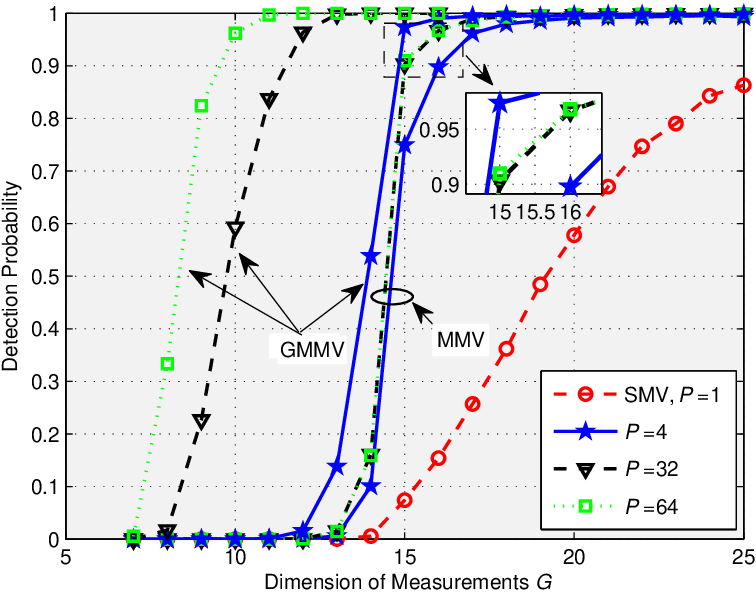}
\end{center}
%%\vspace{-10mm}
\caption{Comparison of the sparse signal detection probabilities of the SMV, MMV and
 the proposed GMMV as functions of $G$.}
\label{fig:detect_prob} % Fig 7
%%\vspace{-1mm}
\end{figure}

 We defined the sparse signal detection probability as the probability of correctly
 acquiring the support set of sparse signal (channel). Fig.~\ref{fig:detect_prob} compares
 the detection probabilities as functions of the measurement dimension $G$ achieved by the
 SMV, MMV, and GMMV in noiseless scenario. In the simulation, the length of multiple sparse
 signals was $M=128$ with the common sparsity level $S_a=8$, and the DSAMP algorithm was
 employed to recover sparse signals. In particular, the SMV recovers single sparse signal
 from single measurement matrix, and the MMV jointly recovers $P$ sparse signals with the
 multiple identical measurement matrices, where the elements of the measurement matrix obey
 the i.i.d. ${\cal{CN}}(0,1)$. By contrast, the GMMV recovers $P$ sparse signals with
 mutually independent measurement matrices in parallel, where the elements of the measurement
 matrices also obey the i.i.d. ${\cal{CN}}(0,1)$. From Fig.~\ref{fig:detect_prob}, it is
 clear that the joint processing of multiple sparse signals with the common support set and
 diversifying measurement matrices significantly enhance the performance of sparse signal
 recovery. For example, to obtain the detection probability of one with $P=64$, the MMV
 requires $G=17$, but the proposed GMMV only needs $G = 11$, which indicates a
 reduction of approximately 35\% in the required time slot overhead. Even the GMMV with
 $P=4$ outperforms the MMV with $P=64$.

\begin{figure}[!tp]
%%\vspace{-6mm}
\begin{center}
\includegraphics[width=1\columnwidth, keepaspectratio]{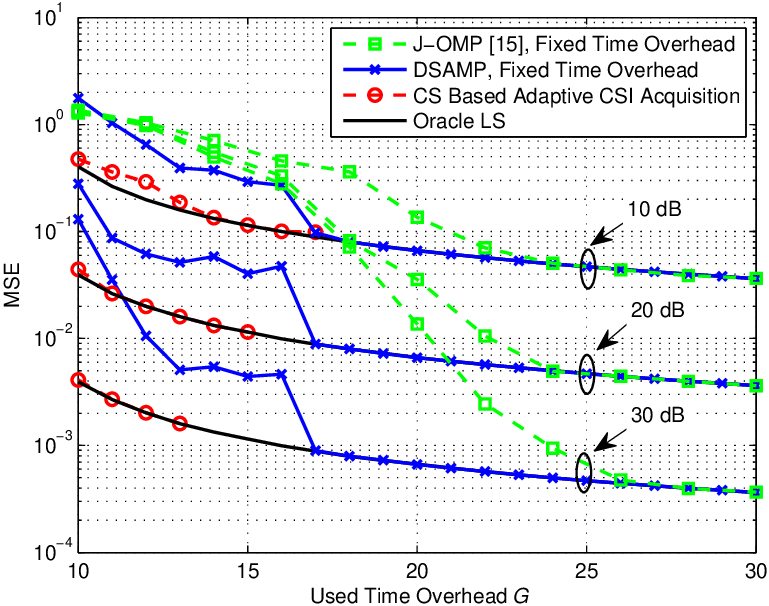}
\end{center}
%%\vspace{-10mm}
\caption{MSE performance of different channel estimation and feedback schemes as functions of the time
 overhead $G$ and SNR.}
\label{fig:mse_vs_T} % Fig 8
%%\vspace{-1mm}
\end{figure}

\begin{figure*}[!hb]
%%\vspace{-4mm}
\begin{center}
\includegraphics[width=18cm, keepaspectratio]{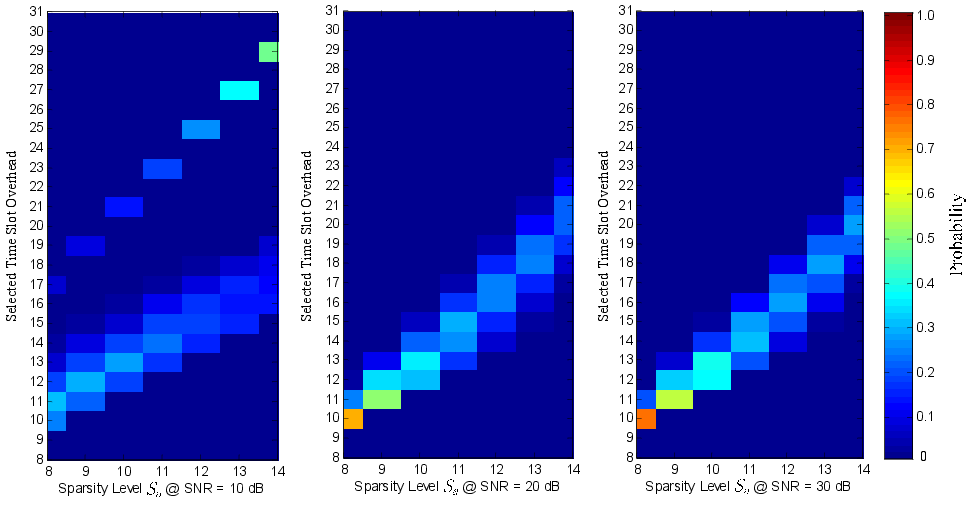}%\columnwidth
\end{center}
%%\vspace{-12mm}
\caption{Distributions of adaptively selected time slot overhead by the CS based adaptive
 CSI acquisition scheme for different sparsity levels and SNRs.}
 \label{fig:dis_pro} % Fig 9
%%\vspace{-4mm}
\end{figure*}
 Fig.~\ref{fig:mse_vs_T} compares the MSE performance of the J-OMP scheme
 \cite{Rao1} with fixed $G$, the DSAMP algorithm with fixed $G$, and the CS based adaptive
 CSI acquisition scheme (Algorithm~\ref{alg1}), where $S_a=8$ was considered. The oracle LS
 estimator with the known support set of the sparse channel vectors was adopted as the
 performance bound. From Fig.~\ref{fig:mse_vs_T}, it can be seen that the J-OMP based CSI
 acquisition scheme performs poorly. By contrast, the proposed DSAMP algorithm is capable
 of approaching the oracle LS performance bound when $G\!>\!2S_a$. However, there still exists
 a significant performance gap between the DSAMP algorithm and the oracle LS estimator for
 $G\!\le\! 2S_a$. This is because the unreliable sparse signal recovery may occur when the
 time slot overhead $G$ is insufficient, which degrades the MSE performance. Fortunately,
 the proposed CS based adaptive CSI acquisition scheme can adaptively adjust $G$ to acquire
 the robust channel estimation. Observe from Fig.~\ref{fig:mse_vs_T} that the proposed CS
 based adaptive CSI acquisition scheme approaches the oracle LS performance bound even for
 $G\!\le\! 2S_a$. Note that for Algorithm~\ref{alg1}, we only plot the MSE associated with
 $G\!\le\! 2S_a$, because Algorithm~\ref{alg1} actually determines an appropriate $G\le 2S_a$
 adaptively.

 Fig.~\ref{fig:dis_pro} depicts the distributions of the adaptively determined time slot
 overhead $G$ by the CS based adaptive CSI acquisition, given different sparsity level $S_a$
 and SNRs. In Algorithm~\ref{alg1}, $G_0$ was set to 8. The results of Fig.~\ref{fig:dis_pro}
 show that the proposed scheme can adaptively determine an appropriate $G$ according to $S_a$.
 As pointed out in Section~\ref{challenging}, to reliably acquire
 CSI, the required $G$ in conventional schemes can be as large as $G=M=128$. By exploiting the
 spatially common sparsity and temporal correlation of massive MIMO channels, the proposed scheme
 can effectively estimate the channels associated with hundreds of antennas at the BS with a
 dramatically reduced time slot overhead. Considering $S_a=8$ at $\mbox{SNR}=30$\,dB for
 example, our scheme only uses a time slot overhead of $G\approx 10$ to acquire CSI at the BS,
 which represents a reduction in the required $G$  by about 92\%, compared to conventional
 schemes.

\begin{figure*}[!t]
%%\vspace{-1mm}
\begin{center}
 \includegraphics[width=18cm, keepaspectratio]{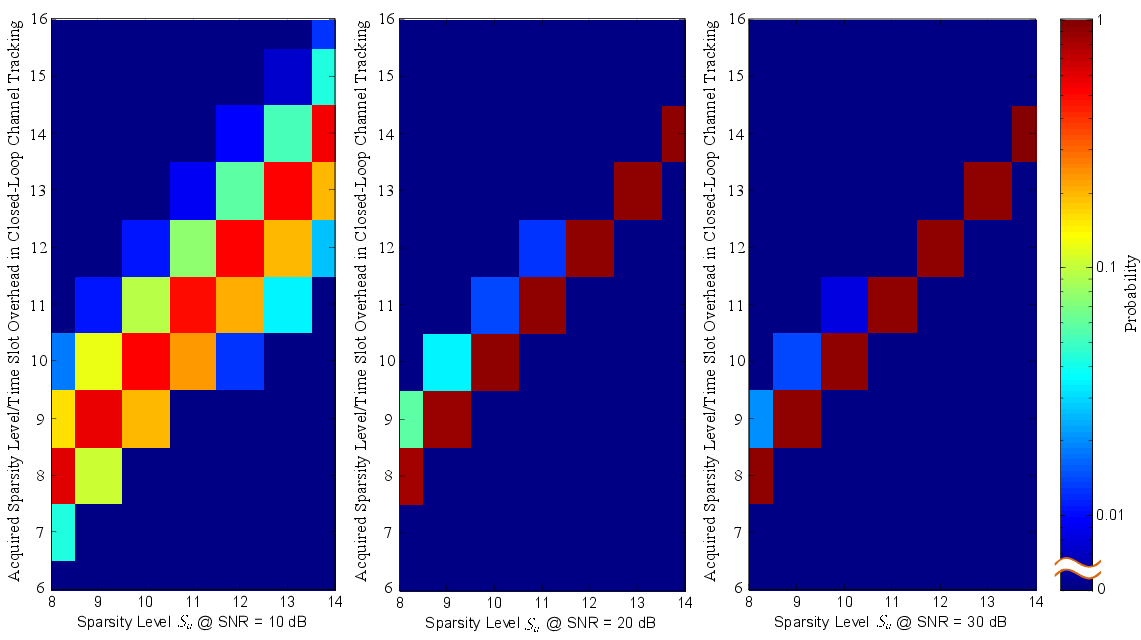}
\end{center}
%%\vspace{-10mm}
 \caption{Distributions of the acquired sparsity level $\widehat{S}_a$ by the CS
 based adaptive CSI acquisition scheme (which is then used as the time slot overhead for the
 proposed closed-loop channel tracking scheme) for different sparsity levels and SNRs.}
\label{fig:dis_pro_loop} % Fig 10
%%\vspace{-7mm}
\end{figure*}

 Fig.~\ref{fig:dis_pro_loop} plots the distributions of the acquired sparsity level
 $\widehat{S}_a$ by the proposed CS based adaptive CSI acquisition scheme, under the same
 settings of Fig.~\ref{fig:dis_pro}. The results of Fig.~\ref{fig:dis_pro_loop} show that
 the proposed scheme can accurately acquire the true sparsity level $S_a$. Note that the
 acquired $\widehat{S}_a$ may be smaller than $S_a$ at low SNR. This is because some virtual
 angular domain coordinates whose channel energy is lower than the noise floor may be
 discarded by the DSAMP algorithm. Because we set the time slot overhead $G$ to $\widehat{S}_a$
 in the closed-loop channel tracking, Fig.~\ref{fig:dis_pro_loop} also provides the probability
 distributions of the time slot overhead used in the closed-loop channel tracking stage. As
 expected, the required time slot overhead in this stage is smaller than the time slot
 overhead actually used in the CS based adaptive CSI acquisition stage, which is confirmed by
 comparing Fig.~\ref{fig:dis_pro_loop} to~Fig.~\ref{fig:dis_pro}.

\begin{figure}[!b]
%%\vspace{-6mm}
\begin{center}
\includegraphics[width=1\columnwidth, keepaspectratio]{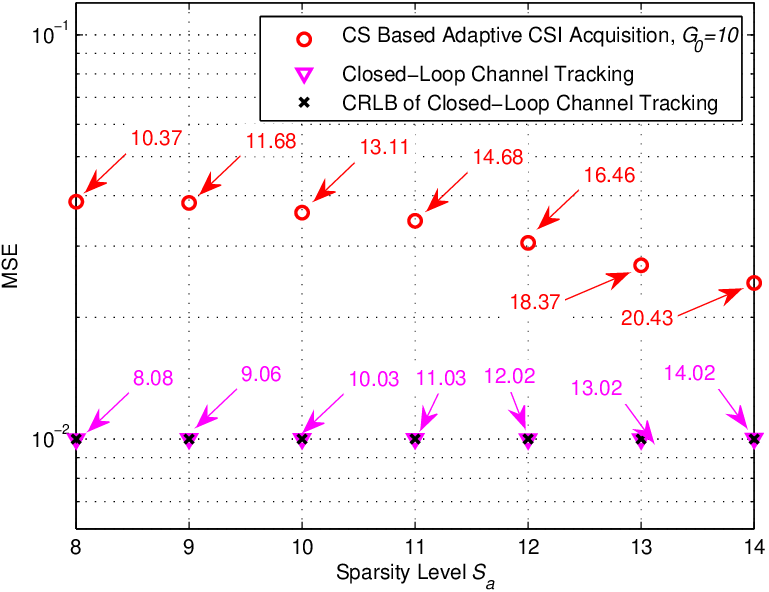}
\end{center}
%%\vspace{-12mm}
\caption{MSE performance comparison of the CS based adaptive CSI acquisition stage and
 closed-loop channel tracking stage for different sparsity levels $S_a$ at $\mbox{SNR}=20$\,dB,
 where the required  $\bar{G}$ for each case is indicated.}
\label{fig:mse_vs_snr2} % Fig. 11
%%\vspace{-1mm}
\end{figure}

 Fig.~\ref{fig:mse_vs_snr2} compares the MSE and required average time slot overhead
 $\bar{G}$ of the CS based adaptive CSI acquisition with those of the closed-loop
 channel tracking for different $S_a$ at $\mbox{SNR}\!=\!20$\,dB.  The initial overhead
 $G_0\!=\!10$ was set for the CS based adaptive CSI acquisition. It is clear that
 benefiting from the accurately estimated sparsity level information provided by the CS
 based adaptive CSI acquisition, the closed-loop channel tracking enjoys the better MSE
 performance with a smaller required time slot overhead. For $S_a=14$, the required
 $\bar{G}$ by the CS based adaptive CSI acquisition and following closed-loop channel
 tracking are 20.43 and 14.02, respectively. Since the acquired CSI by the CS based
 adaptive CSI acquisition is utilized to adaptively adjust the pilot signal for enhancing
 performance, the closed-loop channel tracking approaches the CRLB, as can be seen in
 Fig.~\ref{fig:mse_vs_snr2}. Also note that for the CS based adaptive CSI acquisition, the
 ratio $\bar{G}/S_a$  increases slightly as $S_a$ increases. Hence the MSE performance of
 the CS based adaptive CSI acquisition improves slightly as the true sparsity level $S_a$
 increases.

\begin{figure}[!tp]
%%\vspace{-6mm}
\begin{center}
\includegraphics[width=1\columnwidth, keepaspectratio]{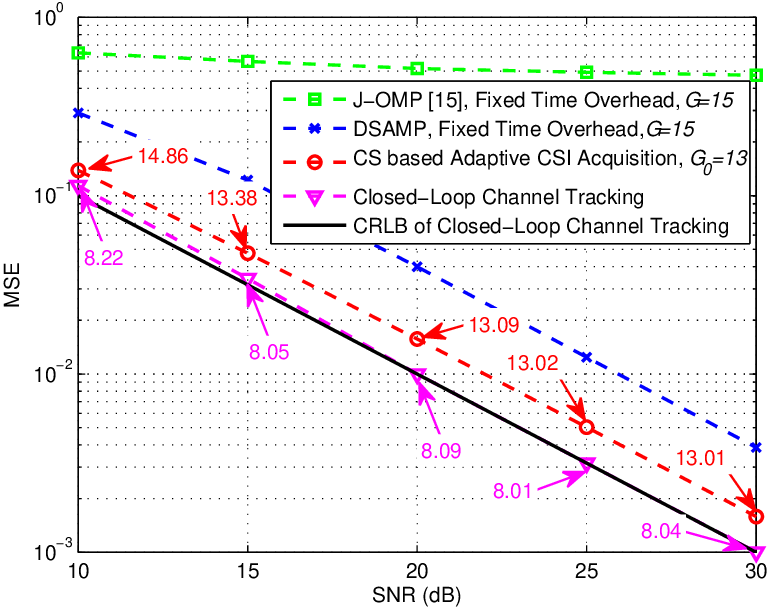}
\end{center}
%%\vspace{-12mm}
\caption{MSE performance comparison of different channel estimation and feedback schemes for various SNRs
 and true sparsity level $S_a=8$, where the required average time slot overhead $\bar{G}$ for
 each case is indicated.}
\label{fig:mse_vs_snr} % Fig. 12
%%\vspace{-1mm}
\end{figure}

 Fig.~\ref{fig:mse_vs_snr} provides the MSE performance comparison for different channel
 estimation and feedback schemes, given $S_a=8$ and various SNRs. Both the J-OMP based CSI
 acquisition scheme \cite{Rao1} and the DSAMP algorithm used the fixed $G\!=\!15$. For the
 CS based adaptive CSI acquisition scheme, $G_0\!=\!13$ was considered. The required average
 time slot overheads for the proposed scheme  are also marked in Fig.~\ref{fig:mse_vs_snr}.
 Again, it is clear that the proposed CS based adaptive CSI acquisition stage
 (Algorithm~\ref{alg1}), which uses the DSAMP algorithm with fixed $G$ to adaptively
 determine an appropriate time slot overhead, outperforms the J-OMP based CSI acquisition
 scheme and DSAMP algorithm with a reduced time slot overhead requirement. By utilizing the
 accurately estimated channel sparsity information provided by the CS based adaptive CSI
 acquisition scheme, the closed-loop channel tracking stage can adaptively adjust the pilot
 signal to approach the CRLB with a further reduced time slot overhead. Specifically, the
 proposed scheme can reliably acquire the CSI of this massive MIMO system, approaching the
 CRLB, with an average time slot overhead $\bar{G}<2S_a$.

\begin{figure}[!tp]
%%\vspace{-2mm}
\begin{center}
\includegraphics[width=1\columnwidth, keepaspectratio]{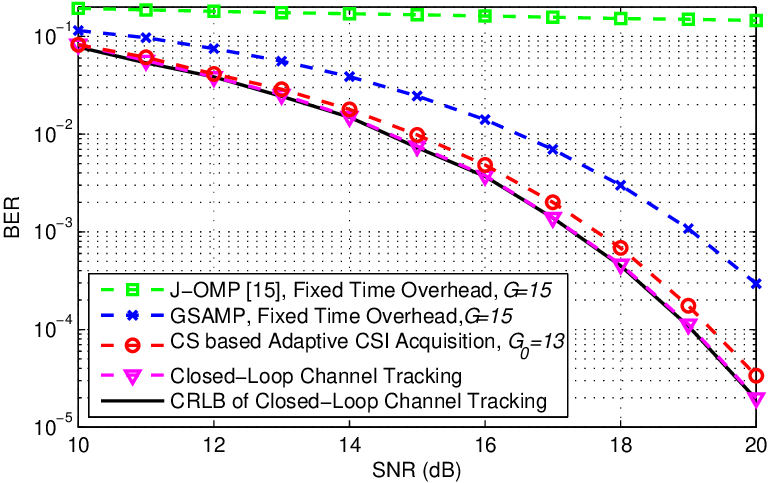}
\end{center}
%%\vspace{-12mm}
\caption{Downlink BER performance with ZF precoding, where the CSI at the BS is
 acquired by different channel estimation and feedback schemes.}
\label{fig:ber_vs_snr} % Fig. 13
%%\vspace{-4mm}
\end{figure}

 Fig.~\ref{fig:ber_vs_snr} compares the downlink bit error rate (BER) performance
 with zero-forcing (ZF) precoding, where the precoding is based on the estimated CSI
 corresponding to Fig.~\ref{fig:mse_vs_snr} under the same setup. In the simulations, the BS
 simultaneously served 16 users using 16-quadrature amplitude modulation signaling, and the
 effective noise in CSI acquisition was only introduced in the downlink channel. It can be
 observed that the proposed channel estimation and feedback scheme outperforms its counterparts,
 and its BER performance is capable of approaching that of the CRLB.

%%\vspace{-1mm}
\section{Conclusions}\label{S6}

 An adaptive channel estimation and feedback scheme has been proposed for FDD massive
 MIMO, which achieves robust and accurate CSI acquisition at the BS, while  dramatically
 reducing the overhead for channel estimation and feedback. The proposed scheme consists
 of two stages, the CS based adaptive CSI acquisition and the following closed-loop
 channel tracking. By exploiting the spatially common sparsity of massive MIMO channels
 within the system bandwidth, the CS based adaptive CSI acquisition can acquire
 the high-dimensional CSI from a small number of non-orthogonal pilots. The closed-loop
 channel tracking, which exploits the spatially common sparsity of massive MIMO channels
 over multiple consecutive time blocks, can effectively utilize the acquired CSI  in the
 first stage to approach the CRLB. We have generalized the MMV to the GMMV
 in CS theory and provided the CRLB of the proposed scheme, which enlightens us to design
 the non-orthogonal pilot for different stages of the proposed scheme.  Simulation results
 have confirmed that our scheme can reliably acquire the CSI of massive MIMO systems,
 specifically, approaching the performance bound  with an adaptively determined time slot
 overhead.

%%\vspace{-1mm}

\end{document}